\documentclass[%
 reprint,
superscriptaddress,
 amsmath,amssymb,
 aps,
]{revtex4-2}

\usepackage{graphicx}
\usepackage{dcolumn}
\usepackage{bm}
\usepackage{braket}
\usepackage{hyperref}
\usepackage{cleveref}
\usepackage{float}
\makeatletter
\let\newfloat\newfloat@ltx
\makeatother
\usepackage{algorithm}
\usepackage{algorithmic}
\usepackage{amsthm}
\usepackage{qcircuit}
\usepackage{subfigure}
\usepackage{color,xcolor,colortbl}
\usepackage{pifont}
\usepackage{multirow}

\theoremstyle{plain}
\newtheorem{theorem}{Theorem}[section]
\newtheorem{lemma}[theorem]{Lemma}

\newtheorem{cor}[theorem]{Corollary}
\theoremstyle{definition}

\newtheorem{rmk}[theorem]{Remark}

\theoremstyle{definition}

\theoremstyle{remark}

\newcommand{\eps}{\epsilon}
\newcommand{\res}{\text{res}}
\newcommand{\La}{\Lambda}
\newcommand{\HR}{\text{HR}}
\renewcommand{\i}{\mathrm{i}}
\renewcommand{\tt}{T_{\mathrm{total}}}
\newcommand{\tm}{T_{\mathrm{max}}}
\newcommand{\tdm}{\widetilde{M}}
\newcommand{\half}{\frac{1}{2}}
\newcommand{\mo}[1]{O\left(#1\right)}
\newcommand{\tmo}[1]{\tilde{{O}}\left(#1\right)}
\newcommand{\ZZ}{\mathbb{Z}}
\newcommand{\RR}{\mathbb{R}}
\renewcommand{\Re}{\mathrm{Re}}
\renewcommand{\Im}{\mathrm{Im}}
\newcommand{\TT}{\mathbb{T}}
\newcommand{\abs}[1]{\left|#1\right|}
\newcommand{\norms}[1]{\abs{#1}_\TT}
\newcommand{\rd}[1]{[#1]}
\DeclareMathOperator{\dist}{dist}
\newcommand{\xiij}{(\theta_i-\theta_j)}
\newcommand{\deij}{(\zeta_i+\zeta_j)}
\newcommand{\p}{\rho}
\newcommand{\PP}{\mathbb{P}}
\newcommand{\cmt}[1]{\textcolor{blue}{#1}}
\newcommand{\xmark}{\textcolor{red}{\ding{55}}}
\newcommand{\cmark}{\textcolor{green}{\ding{51}}}

\begin{document}

\preprint{APS/123-QED}

\title[On adaptive low-depth quantum algorithms for robust multiple-phase estimation]{On adaptive low-depth quantum algorithms for robust multiple-phase estimation}

\author{Haoya Li}
 \email{lihaoya@stanford.edu}
 \affiliation{Department of Mathematics, Stanford University, Stanford, CA 94305}
\author{Hongkang Ni}%
\email{hongkang@stanford.edu}
\affiliation{Institute for Computational and Mathematical Engineering, Stanford University, Stanford, CA 94305}


\author{Lexing Ying}
\email{lexing@stanford.edu}
\affiliation{Department of Mathematics, Stanford University, Stanford, CA 94305}
\affiliation{Institute for Computational and Mathematical Engineering, Stanford University, Stanford, CA 94305}

\date{\today}

\begin{abstract}
This paper is an algorithmic study of quantum phase estimation with multiple eigenvalues. We present robust multiple-phase estimation (RMPE) algorithms with Heisenberg-limited scaling. The proposed algorithms improve significantly from the idea of single-phase estimation methods by combining carefully designed signal processing routines and an adaptive determination of runtime amplifying factors. They address both the {\em integer-power} model, where the unitary $U$ is given as a black box with only integer runtime accessible, and the {\em real-power} model, where $U$ is defined through a Hamiltonian $H$ by $U = \exp(-2\pi\i H)$ with any real runtime allowed. These algorithms are particularly suitable for early fault-tolerant quantum computers in the following senses: (1) a minimal number of ancilla qubits are used, (2) an imperfect initial state with a significant residual is allowed, (3) the prefactor in the maximum runtime can be arbitrarily small given that the residual is sufficiently small and a gap among the dominant eigenvalues is known in advance. Even if the eigenvalue gap does not exist, the proposed RMPE algorithms can achieve the Heisenberg limit while maintaining (1) and (2). 
\end{abstract}

\maketitle


\section{Introduction}\label{sec:intro}

Quantum phase estimation (QPE) is a fundamental problem in quantum computing. In this paper, we focus on the more complex scenario where there are multiple eigenvalues to be estimated. This problem is essentially different from the estimation of a single eigenvalue in many aspects, and most of the methods for single-mode QPE cannot be directly extended to the multiple-mode case. For example, if one directly applies Kitaev's algorithm or the robust phase estimation (RPE) algorithm (\cite{kimmel2015robust,belliardo2020achieving,russo2021evaluating}), which are both commonly accepted benchmarks for single-mode QPE, then different eigenmodes lead to an aliasing effect such that the precision of the estimations cannot be improved in the iterative procedure. In this paper, we present an ensemble of adaptive methods using low-depth circuits that solve this problem successfully. The multiple-phase estimation problem also becomes particularly challenging when the gap between the dominant eigenvalues gets smaller. Our method also handles this issue effectively by incorporating a simple line spectrum estimation algorithm, which we developed in \cite{li2023note} that alleviates the constraint on spectral gaps. 

A few key metrics are needed for evaluating the performance of multiple-phase algorithms. For example, one clearly prefers an algorithm with a small number of qubits and a low circuit depth while allowing for a high level of residual in the initial state. We elaborate on this afterward and show that the adaptive method we propose improves upon existing methods in terms of these metrics, making it well-suited for early fault-tolerant quantum computers.

\subsection{Problem settings}
Formally, this paper concerns the quantum phase estimation (QPE) problem with multiple eigenvalues as follows. For a unitary matrix $U$, let $\{(e^{-2\pi\i\lambda_s}, \ket{\psi_s})\}$ be the eigenpairs of $U$ with  $\lambda_s \in [0,1]$. Suppose that $\ket{\psi}$ is an initial quantum state of the form 
\begin{equation}\label{eq:psi}
\ket{\psi}=\sum_{s=1}^{S}c_s \ket{\psi_s} + c_\res \ket{\psi_\res},
\end{equation}
where $S$ is the number of {\em dominant} eigenvalues and $c_\res \ket{\psi_\res}$ is the {\em residual}. Here, by dominant, we mean that the overlaps between these eigenstates $\ket{\psi_s}$ and $\ket{\psi}$ are bounded from below by a constant $\beta$, i.e., $\min_{1\le s\le S} |c_s|^2\ge\beta>0$. On the other hand, the energy of the residual $|c_\res|^2$ is bounded from above by a constant $\omega$ less than $\beta$, i.e., $|c_\res|^2\le \omega<\beta$. The goal is to estimate the set $\La\equiv\{\lambda_s\}_{s=1}^S$ of dominant eigenvalues up to a prescribed accuracy $\eps$. 

\subsection{Key metrics for performance evaluation}
The difficulty level of multiple eigenvalue estimation depends on the gap between the dominant eigenvalues. If the gap is small, the problem becomes difficult. 
When the spectral gap is bounded from below by a constant independent of the desired precision $\eps$, we refer to it as the {\em gapped} case. When no lower bound of the gap is assumed, we refer to it as the {\em gapless} case. 

Several key complexity metrics to assess a multiple-phase estimation algorithm are listed as follows:
\begin{enumerate}
\item The number of ancilla qubits required. The smaller, the better.
\item The amount of residual $\omega$ allowed in the initial state $\ket{\psi}$, i.e., $\omega$ is the maximum $|c_\res|^2$ allowed.
\item The maximum runtime $\tm$. It is defined as the maximum depth of the quantum circuits used by the algorithm.
\item The total runtime $\tt$, i.e., the sum of the circuit depths over all executions. It has been shown in \cite{giovannetti2006quantum, zwierz2010general, zhou2018achieving}, for example, that $\tt$ has a lower bound named the Heisenberg limit $\tt=\Omega(\eps^{-1})$.
\item Finally, the minimum gap between the eigenvalues allowed.
\end{enumerate}
Among these metrics, a small $\tm$ is particularly important for early fault-tolerant quantum devices since these devices typically have a relatively short coherence time. 

Based on these metrics, an ideal phase estimation algorithm should meet the following requirements:
\begin{enumerate}
\item Using a small number of (even a single) ancilla qubits. \label{en:req1}
\item Allowing the initial state $\ket{\psi}$ to be inexact and ideally $\omega$ to be proportional or even close to $\beta$. \label{en:req2}
\item Satisfying $\tm = \mo{\eps^{-1}}$ with the prefactor ideally proportional to $\omega/\beta$. \label{en:req3}
\item Achieving the Heisenberg-limited scaling $\tt = \tmo{\eps^{-1}}$, where $\tilde{O}$ means omitting the poly-logarithmic term. \label{en:req4}
\item Allowing the minimum gap to be arbitrarily small. \label{en:req5}
\end{enumerate}

\subsection{The integer-power model and the real-power model}
The design of multiple eigenvalue estimation algorithms depends closely on the representation of $U$. In the most general model, $U$ is represented by a quantum circuit or even a black box model, such as a quantum approximate optimization algorithm (QAOA) or variational quantum eigensolver (VQE). This model only allows access to integer powers $U^j$ of $U$ for $j\in\ZZ$, $j\ge0$, and we refer to it as the {\em integer-power} model.  A different model is $U = e^{-2\pi\i H}$ where $H$ is a quantum Hamiltonian with eigenvalues in $[0,1]$. Here, $U$ is often implemented with Trotterization or the splitting method, where one of the main applications is to find the energy of the ground state and a few low-lying excited states of $H$. This model allows access to $U^t:=e^{-2\pi\i H t}$ for any $t\in\RR_+$, and we refer to it as the {\em real-power} model. 


\subsection{Related work} As a fundamental primitive of quantum algorithms, quantum phase estimation has attracted a lot of research activities in the past few decades.

\subsubsection{Single eigenvalue}
Several existing methods can be applied when there is only one dominant eigenvalue, i.e., $S=1$. The early algorithms require a perfect eigenstate, i.e., $c_\res=0$. One of the most fundamental algorithms is the Hadamard test, as shown in \Cref{fig:circuits}(a), which utilizes the $I$ and $S$ gates after the controlled-$U$ gate. The Hadamard test provides estimations of $\Re{\braket{\psi|U|\psi}}$ and $\Im{\braket{\psi|U|\psi}}$, respectively, from the probability of getting $\ket{0}$ while measuring the ancilla qubit. For the Hadamard test, one needs $\mo{\eps^{-2}}$ repetitions to reach precision $\eps$, leading to a total gate complexity $\tt=\mo{\eps^{-2}}$.

A quadratic improvement proposed by Kitaev \cite{kitaev1995quantum,kitaev2002classical} uses measurements of $\braket{\psi|U^{2^j}|\psi}$ for $j = 0, 1, 2, \ldots, J$ with $J=\mo{\log(\eps^{-1})}$, as illustrated by the circuit in \Cref{fig:circuits}(b). The total runtime $\tt = \mo{\eps^{-1}}$ of Kitaev's method achieves the Heisenberg limit (\cite{giovannetti2006quantum, zwierz2010general, zhou2018achieving}). However, the original version of Kitaev's method only applies to perfect eigenstates, i.e., $c_\res=0$.

Another example that reaches the Heisenberg limit is the QPE algorithm with quantum Fourier transform (QFT) \cite{cleve1998quantum}, which only involves a single execution but needs more ancilla qubits and a deeper circuit. Many alternatives have also been proposed in the recent literature \cite{berry2015simulating,higgins2007entanglement,knill2007optimal,poulin2009sampling,o2019quantum,dong2022ground,lin2020near,lin2022heisenberg,ding2022even}. For a more comprehensive overview about the QPE algorithms for a single eigenvalue, we refer to the detailed discussions in \cite{ding2022even,lin2022heisenberg,nielsen2001quantum}.

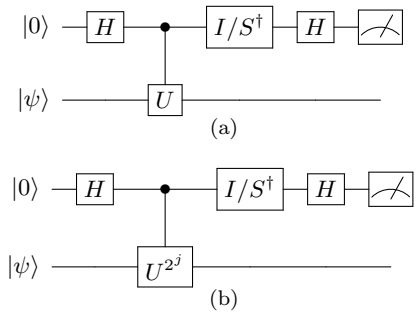
\begin{figure}[ht]
  \centering
   \subfigure[ ]{
    \Qcircuit @C=1em @R=1.5em {
\lstick{\ket{0}} & \gate{H} & \ctrl{1}  & \gate{I/S^\dagger}& \gate{H} &\meter  \\
\lstick{\ket{\psi}} & \qw & \gate{U}  & \qw& \qw &\qw \\
} }\hspace{3em}
\subfigure[ ]{
\Qcircuit @C=1em @R=1.5em {
\lstick{\ket{0}} & \gate{H} & \ctrl{1}  & \gate{I/S^\dagger}& \gate{H} &\meter  \\
\lstick{\ket{\psi}} & \qw & \gate{U^{2^j}}  & \qw& \qw &\qw \\
}}
\caption{(a) Illustration for the Hadamard test. Here $\mathrm{H}$ denotes the Hadamard gate. When estimating the real part of $\braket{\psi|U|\psi}$, we use $I$ (the identity) for the gate after the controlled-$U$ gate, while $S^+$ (the Hermitian conjugate of the phase gate $S$) is used for the estimation of the imaginary part. (b) Illustration for the Kitaev algorithm. For a sequence of powers $j$, the real and imaginary parts of $\braket{\psi|U^{2^j}|\psi}$ are estimated to reach a higher precision and improve the total complexity.}
\label{fig:circuits}
\end{figure}

As we mentioned earlier, for early fault-tolerant quantum devices, besides using a small number of (or even only a single) qubits and reaching the Heisenberg limit $\tt=\tmo{\eps^{-1}}$ for the case $S=1$ and $|c_\res|^2>0$, it is also desired to allow the prefactor in $\tm=\mo{\eps^{-1}}$ to depend on the magnitude of $c_\res$. The authors of \cite{ding2022even} proposed QCELS, an optimization subroutine that works with $|c_\res|^2\le0.29$ and allows the prefactor of $\tm = \mo{\eps^{-1}}$ to scale as $\Theta(|c_\res|)$. In a recent work \cite{ni2023kitaev}, we showed that a modified version of the robust phase estimation (RPE) algorithm \cite{kimmel2015robust,belliardo2020achieving,russo2021evaluating} 
can work with $|c_\res|^2\le 2\sqrt{3}-3\approx 0.464$ and it gives the near-optimal prefactor $\Theta(|c_\res|^2)$.


\subsubsection{Multiple eigenvalues}
The work \cite{o2019quantum} considers the problem of estimating multiple eigenvalues with a signal processing subroutine and adopts the matrix pencil method. The method can be sensitive to noise, and the Heisenberg-limited scaling is not achieved. The algorithm proposed in \cite{somma2019quantum} also estimates multiple eigenvalues based on time series analysis. The total cost is $\mo{\eps^{-6}}$, which is quite far from the Heisenberg limit.  A more recent work \cite{dutkiewicz2022heisenberg} extends the idea of the RPE algorithm to multiple eigenvalues and achieves the Heisenberg limit. However, the residual term $\ket{\psi_\res}$ is not allowed in \cite{dutkiewicz2022heisenberg}, and the prefactor power of $S$ is quite large.

The work \cite{ding2023sim} extended the QCELS algorithm to the multiple eigenvalue setting. This extension achieves the Heisenberg limit with a single ancilla qubit and allows for a residual. The maximum circuit depth can also be reduced when the residual amount $\omega$ is small compared to $\beta$. However, the performance of this optimization technique relies on the existence of a spectral gap $\Delta$ between the multiple dominant eigenvalues, where the minimal runtime of the circuits is $\Omega(\Delta^{-1})$. Moreover, a grid search is generally needed to solve the optimization problem due to the highly complicated landscape, which leads to a classical cost exponential in $S$. 

Quantum subspace diagonalization methods have also been used for multiple eigenvalue estimation \cite{mcclean2017hybrid, motta2020determining, cortes2022quantum,klymko2022real,parrish2019quantum,seki2021quantum}, where the eigenvalues are obtained by solving certain projected eigenvalue problems or singular value problems. Compared to the numerical performance demonstrated, the theoretical analysis of the subspace diagonalization methods is still rather preliminary and pessimistic \cite{epperly2022theory}.

\subsection{Contributions}

This paper conducts a comprehensive study of the multiple eigenvalue quantum phase estimation. Our main contribution is a family of robust multiple-phase estimation (RMPE) algorithms for both the gapless and the gapped cases and for both the real-power and integer-power models. These RMPE algorithms build on the overall structure of the method proposed in \cite{dutkiewicz2022heisenberg}.

For the gapless case, the proposed RMPE algorithm estimates the dominant $S\ge1$ eigenvalues for any residual overlap $|c_\res|^2<\beta$. It utilizes the measurement results from Hadamard tests to the unitary operators $U^{M_\ell}$ for an exponentially-growing sequence $\{M_\ell\}$, one for each step $\ell$. The estimated dominant eigenvalues of $U^{M_\ell}$ allow us to narrow down the intervals that contain the dominant eigenvalues of $U$, and the desired precision will be obtained after $\mo{\log(\eps^{-1})}$ steps. One key component is a simple but new signal processing routine that estimates eigenvalue locations from quantum measurements without any gap assumption. Another key component of the algorithm is the careful choice of $\{M_\ell\}$ to avoid potential collision of the intervals. To achieve this for the real-power model, we utilize non-integer $\{M_\ell\}$ that grows each time by a factor between $2$ and $4$. For the integer-power model, we leverage results from prime number theory for the choice of $\{M_\ell\}$ to avoid collisions. Both the real-power and integer-power versions can tolerate a residual $\omega$ up to $\beta$ and achieve the Heisenberg limit. In summary, this algorithm meets all five requirements, except that the prefactor of $\tm$ does not scale like $\omega/\beta$.

For the gapped case, we improve the algorithms with two key modifications. First, we adopt the ESPRIT algorithm \cite{roy1989esprit}, a different signal processing routine that allows for more accurate eigenvalue estimation when a gap is available. This allows us to build fairly accurate approximations to $\La$ even at the initial step. Second, with this better initial estimate, the number of iterations of the algorithm can be significantly reduced, resulting in $\tm$ to scale like $\omega/\beta$. Both the real-power and integer-power versions can tolerate residual $\omega = O(\beta)$ and reach the Heisenberg limit. In summary, this algorithm meets the first four requirements mentioned above.

Finally, for the case of a finite but small spectral gap, we propose hybrid algorithms for both the real-power and integer-power models. The scaling of maximal runtime is improved to $\mo{\Delta^{-1}+\eps^{-1}}$, which is significantly better than the $\mo{\Delta^{-1}\eps^{-1}}$ complexity of the algorithms for gapped case. Meanwhile, we still maintain the prefactor $\omega/\beta$ and the Heisenberg limit. In a nutshell, the hybrid algorithm uses the gapless signal processing subroutine in the first several iterations as a burn-in period and then switches to the ESPRIT after the spectral gap is enlarged. To summarize, the hybrid algorithms meet the first four requirements mentioned above while improving the scaling of $\Delta$ compared to the algorithms for the gapped case.

\cmt{The maximal runtime and total runtime of different algorithms in this paper are summarized in \cref{table:all}.
\begin{table*}[ht]
\renewcommand{\arraystretch}{1.5}
\begin{tabular}{c|c|c|c}
\hline\hline
Algorithm & Setting & $\tm$ & $\tt$ \\ \hline
\Cref{sec:nonint} & gapless, non-integer & $\mo{\frac{1}{\eps}\log\frac{1}{\beta-\omega}}$ & $\tmo{\frac{S^2}{\eps(\beta-\omega)^2}}$ \\ \hline
\Cref{sec:int} & gapless, integer & $\mo{\frac{1}{\eps}\log\frac{1}{\beta-\omega}}$ & $\tmo{\frac{S^5}{\eps(\beta-\omega)^2}}$  \\ \hline
\Cref{sec:nonint2} & gapped, non-integer &$\mo{\frac{(\Delta^{-1}+S^2)S^2\omega}{\beta\eps}}$
    & $\tmo{\frac{(\Delta^{-1}+S^2)S^2}{\omega\beta\eps}}$  \\ \hline
\Cref{sec:int2} & gapped, integer &$\mo{\frac{(\Delta^{-1}+S^4)S^2\omega}{\beta\eps}}$
    & $\tmo{\frac{(\Delta^{-1}+S^4)S^2}{\omega\beta\eps}}$ \\ \hline
\Cref{sec:noint3} & gapped, non-integer & $O_{S,\beta}\left(\max\{\omega^{-1}\Delta^{-1}, \omega\eps^{-1}\}\right)$
    & $\tilde{O}_{S,\beta}\left(\omega^{-2}\Delta^{-1}+ \omega^{-1}\eps^{-1}\right)$  \\ \hline
\Cref{sec:int3} & gapped, integer & $O_{S,\beta}\left(\max\{\omega^{-1}\Delta^{-1}, \omega\eps^{-1}\}\right)$
    & $\tilde{O}_{S,\beta}\left(\omega^{-2}\Delta^{-1}+ \omega^{-1}\eps^{-1}\right)$ \\ \hline\hline
\end{tabular}
\caption{The summary of different algorithms presented in this paper.}
\label{table:all}
\end{table*}
}
\subsubsection{Comparisons}
Compared with \cite{dutkiewicz2022heisenberg}, our RMPE algorithms have the following advantages. First, the proposed algorithms allow an imperfect initial state with a residual term. Second, by using a simpler signal processing routine in both the gapped and gapless cases, we obtain a smaller prefactor in terms of the sparsity level $S$: the prefactor of $\tt$ is only quadratic in terms of the sparsity level $S$. In contrast, the dependency on $S$ obtained in \cite{dutkiewicz2022heisenberg} is $\mo{S^{12}}$. Third, the proposed RMPE algorithms work in the real-power and integer-power models, while \cite{dutkiewicz2022heisenberg} only works with the real-power model. Furthermore, when the dominant eigenvalues present a gap, $\tm$ of our gapped algorithm has a prefactor of order $\omega/\beta$, which is a property not shared by the algorithm of \cite{dutkiewicz2022heisenberg}.

Compared with \cite{ding2023sim}, the proposed RMPE algorithms have two advantages. First, we address the gapless case that is not visited in \cite{ding2023sim}. Second, the classical computational complexity is much lower: our RMPE algorithms scale polynomially in $S$ while the optimization step of the method in \cite{ding2023sim} generally scales exponentially in $S$.

\cmt{For clarity, a comparison of different multi-phase estimation algorithms is presented in \Cref{table:compare}.
\begin{table*}[ht]
\begin{tabular}{c|c|c|c|c|c|c}
\hline\hline
\multirow{2}{*}{Algorithm} & Allow  & Heisenberg  & Gapless  & Integer & Short  & \multirow{2}{*}{Remark} \\ 
& residual & limit & case & power & depth & \\ \hline
Ref. \cite{o2019quantum} & \xmark & \xmark & \xmark & \cmark &?& \\
Ref. \cite{somma2019quantum} & \cmark & \xmark & \cmark & \cmark & \cmark &\\
Ref. \cite{dutkiewicz2022heisenberg} & \xmark & \cmark & \cmark& \xmark & \xmark & Large prefactor in power of $S$\\
Ref. \cite{ding2023sim} & \cmark & \cmark & \xmark& \xmark & \cmark & Large classical computational cost\\
Ref. \cite{mcclean2017hybrid, motta2020determining, cortes2022quantum,klymko2022real,parrish2019quantum,seki2021quantum} & \cmark & ? & ? & ? &?& \\
(This work) Sec. \ref{sec:ns} & \cmark & \cmark & \cmark & \cmark & \xmark & \\
(This work) Sec. \ref{sec:hybrid} & \cmark & \cmark & \xmark & \cmark & \cmark & \\ \hline\hline
\end{tabular}
\caption{The comparison of different multi-phase estimate algorithms.}
\label{table:compare}
\end{table*}
}

\subsection{Contents}

The rest of the paper is organized as follows. \Cref{sec:main} introduces the main structure of the proposed algorithms. \Cref{sec:ns} discusses the gapless case in detail. In \Cref{sec:sp}, we present the easier gapped case while focusing on its difference from the gapless case. In \Cref{sec:hybrid}, we provide the hybrid algorithms for the gapped case. All these sections address both the real-power and integer-power models.

A few comments about the notations are in order here. For a measurable set $T$, we use $\abs{T}$ and $\norms{T}$ to denote the Lebesgue measure of $T\subset \RR$ and $T\subset \TT\equiv\RR/\ZZ$, respectively. The $\eta$-neighborhood of a set $T\subset \RR$, denoted by $B(T,\eta)$, is defined as:
\begin{equation}
  B(T,\eta):=\cup_{t\in T}[t-\eta, t+\eta].
\end{equation}
The $\eta$-neighborhood of a set $T\subset \TT$, denoted as $B_\TT(T,\eta)$, is defined similarly on the torus $\TT$. For a set $T\subset \RR$, we denote the set $\{ct+d|t\in T\}$ for any $c, d\in\RR$ by $cT+d$.

\section{Algorithm outline}\label{sec:main}

This section describes the main algorithmic structure. The following sections specialize this structure to the gapless and gapped cases, respectively.

Applying the Hadamard test to $U^1,\ldots,U^K$ for some integer $K$ with input $\ket{\psi}$ provides a noisy measurement $\{y(k)\}$ of 
\[
\braket{\psi|U^{k}|\psi} = \sum_{s=1}^S |c_s|^2 e^{-2\pi\i \lambda_s k} + |c_\res|^2\braket{\psi_\res|U^{k}|\psi_\res}, 
\]
for $k=-K,\ldots,K$.
This is the Fourier transform of a probability measure $f(x)$ defined as
\[
f(x) \equiv \sum_{s=1}^S |c_s|^2 \delta_{\lambda_s}(x) + \text{Residual},
\]
where $\text{Residual}$ is a positive measure with mass $|c_\res|^2$. $\La=\{\lambda_s\}_{s=1}^S$ can then be viewed as the dominant support of the measure $f(x)$. 

At the first step, from the noisy approximation $y(k)$ to the Fourier coefficient $\hat{f}(k)$, \cmt{which is obtained from averaging the output of $N_\HR$ repetitions of Hadarmard test,} one can first extract a rough estimation $E_0$ of $\La$ using an appropriate signal processing routine. \cmt{Roughly speaking, at this first step, we want to estimate $\La$ to $\eta$ precision. The signal processing routine and the determination of parameter $\eta$ will be detailed in the following sections.}

The algorithm then chooses a sequence of amplifying factors $(m_1, m_2, \ldots, m_\ell )$. Define $M_0=1$ and $M_\ell = m_1m_2\cdots m_\ell $. At the $\ell$-th step, we start with $E_{\ell-1}$, the current approximation of $\La$. Applying the Hadamard test to $U^{M_\ell},\ldots,U^{M_\ell K}$ with input $\ket{\psi}$ gives rise to noisy measurements $\{y_\ell(k)\}$ of
\begin{equation}\label{eq:fhatk}
    \braket{\psi|U^{M_\ell k}|\psi} = \sum_{s=1}^S |c_s|^2e^{-2\pi\i M_\ell\lambda_s k} + |c_\res|^2\braket{\psi_\res|U^{M_\ell k}|\psi_\res}.
\end{equation}
This can be viewed as the Fourier transform of a probability measure $f_\ell(x)$
\[
f_\ell(x) \equiv \sum_{s=1}^S |c_s|^2 \delta_{\rd{\lambda_s M_\ell}}(x) + \text{Residual},
\]
where $\text{Residual}$ is again a positive measure with mass $|c_\res|^2$ and $\rd{x}\equiv(x \mod 1)$. The data $y_\ell(k)$ is a noisy approximation to the Fourier coefficient $\hat{f}_\ell(k)$. Again, with a signal processing subroutine, one can obtain an estimation $Y_\ell$ of the dominant support of $f_\ell(x)$, which is denoted by $\La_\ell=\{\rd{\lambda_s M_\ell}\}_{s=1}^S$. This estimation $Y_\ell$ provides information with an increased resolution and helps improve the current estimation $E_{\ell-1}$ to an update $E_{\ell}$. 

More specifically, with the spectrum estimation $Y_\ell$, one divides it by $M_\ell$ to obtain the set $\frac{1}{M_\ell}Y_\ell$. From the property of the signal processing methods used, $\frac{1}{M_\ell}Y_\ell$ will be a union of at most $S$ disjoint intervals $\frac{1}{M_\ell}Y_\ell = \cup_{i=1}^{S_\ell} I'_{\ell, i}$ with $\La \subset\cup_{i=1}^{S_\ell}(I'_{\ell, i}+\frac{q_{\ell, i}}{M_\ell})$ for properly chosen integers $q_{\ell, i}$. Here $S_\ell\le S$ is the number of intervals in $\frac{1}{M_\ell}Y_\ell$, and $I'_{\ell, i}$ is the $i$-th disjoint interval in $\frac{1}{M_\ell}Y_\ell$. By carefully determining the amplifying factors $(m_1, m_2, \ldots, m_\ell )$, we show that at most one integer $q_{\ell,i}$ satisfies $(I'_{\ell,i}+\frac{q_{\ell,i}}{M_\ell})\cap E_{\ell-1}\not=\varnothing$.  Then we define $(I'_{\ell, i}+\frac{q_{\ell, i}}{M_\ell})$ as $I_{\ell, i}$ and form the updated estimation $E_\ell = \cup_{i=1}^{S_\ell} I_{\ell,i}$. The determination of $m_\ell$ will be discussed in detail in the following sections. Given that $m_\ell$ is chosen properly, the algorithm can be summarized as in \Cref{alg:main}. 

\begin{algorithm}[ht]
c	\caption{Structure of robust multiple-phase estimation (RMPE) algorithm}
	\label{alg:main}
	\begin{algorithmic}
		\STATE{\textbf{Input:} $\eps$: desired precision, $\beta$: the lower bound of dominant eigenvalues, $S$: the number of dominant spikes, $\omega$: upper bound for the residual in the initial state $\ket{\psi}$.}
		\STATE{Set the initial estimation $E_{-1}$.}
            \STATE{Set $M_0 = 1$, $\ell = 0$.}
            \STATE{Calculate $\eta$, $K$, $N_\HR$ and $\alpha$ according to $\eps$, $\beta$, $\omega$ and $S$.}
		\WHILE{$\eta/M_\ell > {\epsilon}$}
            \STATE{Choose an $m_\ell \ge 2$ according to \Cref{lemma:m_real} or \Cref{lemma:primedist} and set $M_\ell= M_{\ell-1}m_\ell $ if $\ell>0$.}
            
            \STATE{Run the circuit in \Cref{fig:circuits}(a) with $U$ replaced by $U^{M_\ell k}$ for $\frac{N_\HR}{2}$ times each for the real and imaginary parts and $0\le k\le K$ and obtain the signal ${y}_\ell$. }
            
            \STATE{Obtain $Y_\ell$ from a signal processing routine, and then several intervals $I'_{\ell,1},\ldots, I'_{\ell, S_\ell }$ from $\frac{1}{M_\ell}Y_\ell$. The choice of $m_\ell $ ensures that for each $I'_{\ell ,i}$, only one integer $q_{\ell ,i}$ gives $(I'_{\ell ,i} + \frac{q_{\ell ,i}}{M_\ell })\cap E_{\ell-1} \ne\varnothing$.}
            
            \STATE{Set $I_{\ell ,i} = I'_{\ell ,i} + \frac{q_{\ell ,i}}{M_\ell }$ and $E_\ell = \bigcup_{i=1}^{S_\ell }I_{\ell ,i}$}
            \STATE{$\ell\leftarrow \ell+1$. }  
            \ENDWHILE
            \STATE{\textbf{Output:} The final estimation $E_L$ as an approximated support of $\La$. }
	\end{algorithmic}
\end{algorithm}

We will show in the next few sections that  
the determination of $K$, $N_\HR$, $m_\ell$ and $\eta$ ensures that $E_\ell = \cup_{i=1}^{S_\ell }I_{\ell ,i}$ enjoys the following properties:
\begin{enumerate}
\item $\{ I_{\ell,i}\} $ are disjoint, and $\sum_{i=1}^{S_\ell }|I_{\ell ,i}|\le S\eta/M_\ell $.\label{en:1}
\item For each $I_{\ell ,i}$, the intersection $I_{\ell ,i}\cap\La\ne\varnothing$.\label{en:2}
\item $\La\subset E_\ell \subset B(\La,\frac{\eta}{M_\ell })$ (real-power model) or $\La\subset E_\ell \subset B_\TT(\La,\frac{\eta}{M_\ell})$ (integer-power model).\label{en:3}
\end{enumerate}
Since the chosen $m_\ell$ satisfies $m_\ell \ge2$, it can be deduced from the last property above that the while loop in \Cref{alg:main} ends in $\mo{\log\left(\frac{\eta}{\eps}\right)}$ iterations with high probability. In \Cref{sec:nonint}, we show that for the real-power model, one can choose a proper $m_\ell \in[2, 4]$ according to the previous estimations. On the other hand, for the integer-power model, one can choose an appropriate $m_\ell $ with the help of prime numbers. Detailed explanations and analyses are provided in \Cref{sec:int}. Compared with the original version of Kitaev's method, where $m_\ell =2$ for all $\ell$, the adaptive calculation of these factors enables the proposed algorithm to address QPE problems with multiple dominant eigenvalues and a non-zero residual. 

Here, we discuss some specific aspects of the signal processing routine used to extract $Y_\ell$. At the $\ell $-th iteration, we implement $N_\HR$ measurements of $\hat{f}_\ell (k)$ for each $k \in\{0, 1, \ldots, K\}$ such that the averaged measurement result ${y}_\ell(k)$ satisfies 
\[
|{y}_\ell(k)-\hat{f}_\ell(k)|\le\alpha
\]
for all $k$ with high probability. The determination of parameters $N_\HR$, $K$, and $\alpha$ will be elaborated in the following sections and is omitted for now. The problem of recovering $\La_\ell=\{\rd{\lambda_s M_\ell} \}_{s=1}^S$ from $\{y_\ell(k)\}_{k=0}^K$ has been extensively studied under the name line spectrum estimation, and plenty of established results for line spectrum estimations can be used. As explained earlier,  \Cref{alg:main} requires that $Y_\ell$ satisfies the following three requirements.
\begin{enumerate}
    \item $Y_\ell$ is a union of at most $S$ disjoint intervals such that each interval contains at least one dominant eigenvalue. \label{en:y1}
    \item $\La_\ell$ is a subset of $Y_\ell$. \label{en:y2}
    \item $|Y_\ell|\le S\eta$ for a parameter $\eta$ to be determined by the algorithm. \label{en:y3}
\end{enumerate}
We provide a detailed description in \Cref{thm:yk} and show that these requirements can indeed be satisfied for both the gapless case (\Cref{subsec:spnogap}) and the gapped case (\Cref{sec:spgap}).

\section{The gapless case}\label{sec:ns}

This section specializes \Cref{alg:main} to the gapless case, i.e., no gap is assumed among the dominant eigenvalues. The signal processing routine proposed in \cite{li2023note} satisfies the requirements \ref{en:y1}, \ref{en:y2} and \ref{en:y3} listed above for $Y_\ell$ in \Cref{sec:main}. After summarizing its main results in \Cref{subsec:spnogap} for completeness, we discuss the real-power model in \Cref{sec:nonint} and the integer-power model in \Cref{sec:int}, respectively. Figure \ref{fig:gapless} gives a graphical illustration of the gapless case algorithm.

\begin{figure*}[ht]
  \centering
  \includegraphics[scale=0.4]{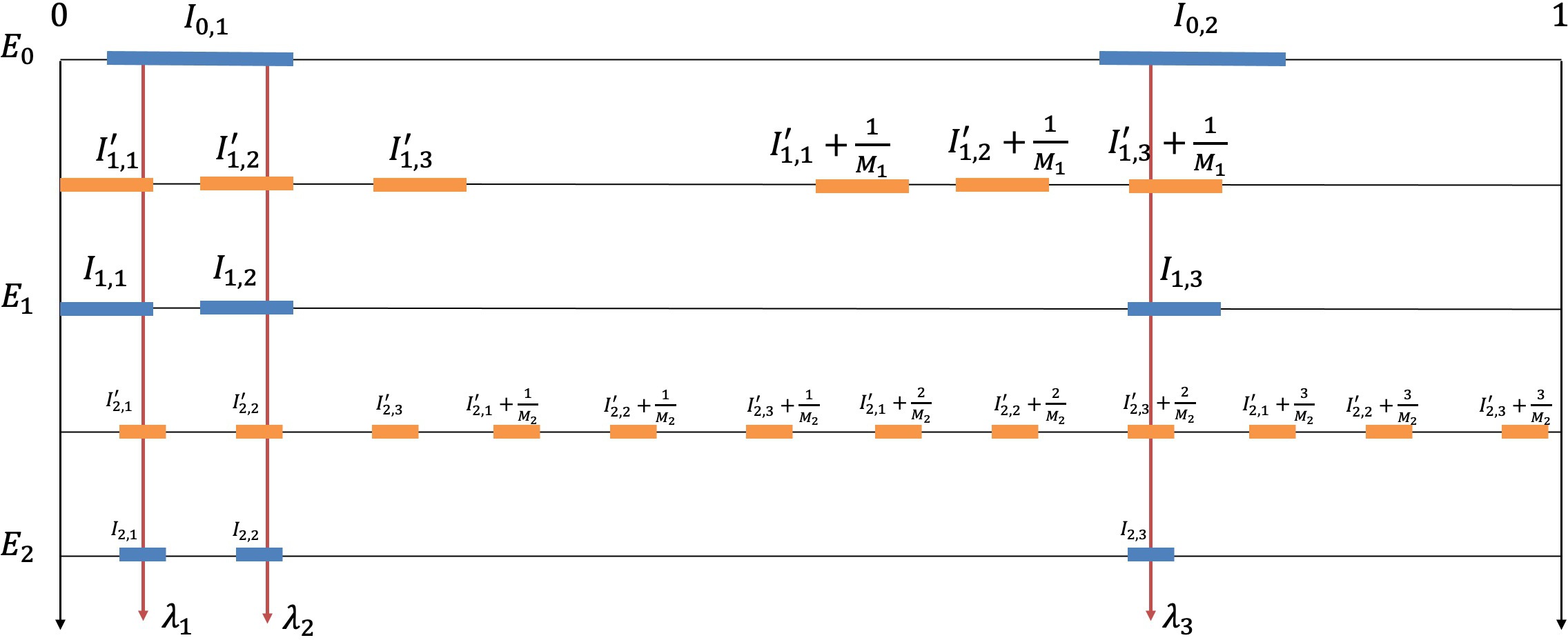} 
  \caption{Illustration of \Cref{alg:main} for the gapless case with $S=3$. At each step, we maintain a union of at most $S$ intervals as an estimation of $\La$, shown by the horizontal blue sticks in the diagram. The yellow intervals represent $M_{\ell}^{-1}Y_\ell$ and its translates by integer multiples of $M_{\ell}^{-1}$. The choice of $M_\ell$ ensures that only one integer $q_{\ell ,i}$ satisfies $I'_{\ell, i} + \frac{q_{\ell , i}}{M_{\ell} }\cap E_{\ell -1}\not=\varnothing$. The union of these $I'_{\ell, i} + \frac{q_{\ell, i}}{M_{\ell}}$ then becomes the new estimation of $\La$ and is used for the next update. Note that the number of intervals in $E_\ell$ may increase. Due to the low resolution, $\lambda_1$ and $\lambda_2$ are covered by the same interval in $E_0$. As the resolution increases, they belong to two different intervals in $E_1$ and $E_2$. }
  \label{fig:gapless}
\end{figure*}

\subsection{Signal processing routine}\label{subsec:spnogap}
Recall that ${y}_\ell$ are the averaged measurements result such that $|{y}_\ell(k)-\hat{f}_\ell(k)|\le \alpha$, where $\hat{f}_\ell$ is defined in \eqref{eq:fhatk}. Here $y_\ell(-k)$ is defined as the complex conjugate of $y_\ell(k)$ for $0\le k\le K$, since $\hat{f}_\ell(-k)$ is the conjugate of $\hat{f}_\ell(k)$. Following \cite{li2023note}, the spikes $\La_\ell$ can be estimated by the following set:
\begin{equation}\label{eq:Iegauss}
  X_\ell = \left\{x:\left|\sum_{|k|\le K}{y}_\ell(k)\hat{\phi}_p(k)e^{2\pi\i k x}\right|>\frac{6\beta+5\omega}{11}\phi_s\right\},
\end{equation}
where $\sigma = \sqrt{\frac{1}{\pi}\log\frac{12}{\beta-\omega}}$,  $\hat{\phi}_p(k) = \exp(-\pi (k \sigma/K)^2),~k \in\ZZ$ and $\phi_s = \sum_{k\in\ZZ}\hat{\phi}_p(k)$. The following theorem holds for this choice of $X_\ell$.

\begin{theorem}\label{thm:spnogap}
  Suppose $\alpha<\frac{\beta-\omega}{3}$ and $K\ge 3\tau$, then $\La_\ell \subset X_\ell$ and $\max_{x\in X_\ell} \dist(x, \La_\ell)\le \tau/K$, where $\tau = \frac{1}{\pi}\log\frac{12}{\beta-\omega}$.
\end{theorem}

Note that it is possible that $X_\ell$ is a disjoint union of more than $S$ intervals, and some of them may not contain a true spike. In the following corollary, we form a set $Y_\ell$ with the help of the estimations $X_\ell$ obtained above to meet the requirements listed in \Cref{sec:main}. The proof can be found in \Cref{sec:ykproof}

\begin{cor}\label{thm:yk}
  Using the set $X_\ell$ we obtained from the above signal processing routine, we can construct a set $Y_\ell$ that satisfies the following properties when $\eta>3\tau/K$:
  \begin{enumerate}
  \item $Y_\ell = \bigcup_{i=1}^{S_\ell}Y_{\ell,i}$ is the disjoint union of intervals, and $|Y_\ell|\le S\eta$.
  \item For each interval $Y_{\ell ,i}$, the intersection $I_{\ell ,i}\cap\La_\ell\ne\varnothing$.
  \item $\La_\ell\subset Y_\ell \subset B_\TT(\La_\ell,\eta)$. 
  \end{enumerate}
\end{cor}

We emphasize that $K$ dictates the estimation accuracy obtained. Even if $\omega$ is zero, the estimation error is proportional to $1/K$. Without making $K$ larger, it is not possible to make the approximation more accurate in this signal-processing routine.

\subsection{The real-power model}\label{sec:nonint}
In this subsection, we aim to recover the eigenvalues of some Hamiltonian $H$ assuming we have access to $e^{-2\pi\i H t}\ket{\psi}$ for $t\in\RR_+$, and $U^t$ is an abbreviation of $e^{-2\pi\i H t}$.

Without loss of generality, we can assume $\La \subset[0, 0.9]$. Otherwise, one can prescale $H$ appropriately. Hence the initialization $E_{-1}$ is defined as $[0, 0.9]$. In this way, the property \ref{en:2} can be guaranteed at step $\ell=0$. Otherwise, if the signal processing subroutine gives an interval $I_{0, i}$ that is $[\gamma_1,\gamma_2]\mod 1$ with $\gamma_1<0$ and $\gamma_2>0$, one cannot tell whether there is a true spike in $[0,\gamma_2]$ or $[1+\gamma_1,1]$. However, this can be avoided under this assumption because the eigenvalues are estimated to error level $\eta<0.1$ at step $\ell=0$, and the spikes near 0.9 and 0 will not interfere with each other.

As previously mentioned, a vital step in \Cref{alg:main} is the determination of $m_\ell $ such that for each $I'_{\ell ,i}$, only one integer $q_{\ell ,i}$ gives $I'_{\ell ,i} + \frac{q_{\ell ,i}}{M_{\ell-1}m_\ell }\cap E_{\ell -1}\not=\varnothing$. Once such $m_\ell$ is chosen, then the properties \ref{en:1}, \ref{en:2}, and \ref{en:3} are satisfied due to \Cref{thm:yk}. 

Now assume that the properties \ref{en:1}, \ref{en:2}, and \ref{en:3} are already satisfied at step $\ell-1$. For any $\lambda_s\in\La$, since $\lambda_s\in E_{\ell-1}$, there must be some $\lambda_s\in I'_{\ell ,i}\mod \frac{1}{M_{\ell-1}m_\ell }$. Thus there exists at least one $q_{\ell ,i}$ such that $I'_{\ell ,i} + \frac{q_{\ell ,i}}{M_{\ell-1}m_\ell }\cap E_{\ell-1}\ne \varnothing$. Choose an arbitrary such $q_{\ell ,i}$ and let $I_{\ell ,i} = I'_{\ell ,i} + \frac{q_{\ell ,i}}{M_{\ell-1}m_\ell }$, then $I_{\ell ,i}\subset B(\La,\frac{\eta}{2M_{\ell-1}})$ as long as $m_\ell \ge 2$. We also have $\La\subset E_{\ell-1}$, which implies $I_{\ell ,i}\subset B(\La,\frac{\eta}{2M_{\ell-1}})\subset B\left(E_{\ell-1},\frac{\eta}{2M_{\ell-1}}\right)$. Therefore, if the choice of $m_\ell $ satisfies 
\begin{equation}\label{eq:aug}
    \left(B\left(E_{\ell-1},\frac{\eta}{2M_{\ell-1}}\right) + \frac{q}{M_{\ell-1}m_\ell }\right)\cap E_{\ell-1} = \varnothing,~ \forall q\in \ZZ\backslash\{0\},
\end{equation} 
then we can deduce that $\left(I_{\ell ,i}+\frac{q}{M_{\ell-1}}\right)\cap E_{\ell-1}=\varnothing$ for any $q\in \ZZ\backslash\{0\}$ and thus the choice of $q_{\ell ,i}$ is unique. Moreover, we have $E_\ell  = \bigcup_{i=1}^{S_\ell } I_{\ell ,i}\subset B(\La,\frac{\eta}{2M_{\ell-1}})\subset B(E_{\ell-1},\frac{\eta}{2M_{\ell-1}})$. Summarizing the discussion above, we obtain the following lemma:
\begin{lemma}\label{lem:inclusion}
    Given the previous amplifying factor $M_{\ell-1}$ and estimation $E_{\ell-1}$ in \Cref{alg:main}, if the choice of $m_\ell $ satisfies \eqref{eq:aug}, then the choice of $q_{\ell ,i}$ in \Cref{alg:main} is unique, and the corresponding $E_\ell $ satisfies
    \begin{equation}
        E_\ell  \subset B\left(\La,\frac{\eta}{2M_{\ell-1}}\right)\subset B\left(E_{\ell-1},\frac{\eta}{2M_{\ell-1}}\right).
    \end{equation}
\end{lemma}

Before establishing \eqref{eq:aug}, we first state the following lemma, whose proof can be found in \Cref{sec:m_realproof}.
\begin{lemma}\label{lemma:m_real}
    Suppose $M$ is a positive number and a set
    \[G = \bigcup_{i=1}^t\left[\theta_i-\frac{\zeta_i}{M},\theta_i+\frac{\zeta_i}{M}\right]
    \]
    is the union of $t$ disjoint intervals with $t\le S$ and $\sum_{i=1}^t\zeta_i\le S\zeta$. If $\zeta\le\frac{1}{8S(2S-1)}$, then there must be some $m\in[2,4]$ such that
    \begin{equation}
        \left[\theta_i-\frac{\zeta_i}{M}-\frac{q}{Mm},\theta_i+\frac{\zeta_i}{M}-\frac{q}{Mm}\right]\cap \left[\theta_j-\frac{\zeta_j}{M},\theta_j+\frac{\zeta_j}{M}\right]=\varnothing\label{eq:disjoint}
    \end{equation}
    holds for all $1\le i,j\le t$, and $q\in\ZZ\backslash \{0\}.$ In other words, it holds that $B(G,\frac{q}{Mm})\cap G = \varnothing$ for any $q\in\ZZ\backslash \{0\}.$
\end{lemma}

Since $E_{\ell-1} = \bigcup_{i=1}^{S_{\ell-1}}I_{\ell -1,i}$ is the union of at most $S$ disjoint intervals and $\abs{E_{\ell-1}} \le \frac{S\eta}{M_{\ell-1}}$, its neighborhood $B\left(E_{\ell-1},\frac{\eta}{2M_{\ell-1}}\right)$ must be $t$ disjoint intervals with $t\le S$ and $\abs{B\left(E_{\ell-1},\frac{\eta}{2M_{\ell-1}}\right)}\le \frac{2S\eta}{M_{\ell-1}}$. This is exactly the case in \Cref{lemma:m_real} when taking $G = B\left(E_{\ell-1},\frac{\eta}{2M_{\ell-1}}\right)$, $M=M_{\ell-1}$ and $\zeta=\eta$. Therefore, the conclusion of \Cref{lemma:m_real} guarantees that we can constructively find an $m_\ell$ such that for any $q\in\ZZ\backslash\{0\}$
\begin{equation*}
    \left(B\left(E_{\ell-1},\frac{\eta}{2M_{\ell-1}}\right) + \frac{q}{M_{\ell-1}m_\ell}\right)\cap B\left(E_{\ell-1},\frac{\eta}{2M_{\ell-1}}\right) = \varnothing,
\end{equation*} 
which is stronger than \eqref{eq:aug}. Thus the requirement on $m_\ell$ in \Cref{alg:main} is satisfied, i.e., for each $I'_{\ell ,i}$, only one integer $q_{\ell ,i}$ gives $(I'_{\ell ,i} + \frac{q_{\ell ,i}}{M_\ell })\cap E_{\ell-1} \ne\varnothing$ due to the choice of $m_\ell$. According to \Cref{lemma:m_real}, one can choose any $\eta\le \frac{1}{8S(2S-1)}$. We thus obtain the following theorem. 

\begin{theorem}\label{thm:complexityreal}
    Define $\tau = \frac{1}{\pi}\log\frac{12}{\beta-\omega}$ and let $\alpha$ be a constant that satisfies $\alpha<\frac{\beta-\omega}{3}$. For any $\eta\le \frac{1}{8S(2S-1)}$, if $K \ge 3\eta^{-1}\tau$    and
    \begin{equation}
       N_\HR = 2\left\lceil\frac{4}{\alpha^2}\left(\log\frac{4}{\p}+\log\left(\left\lceil\log_2\frac{\eta}{\eps}\right\rceil+1\right)+\log(K+1)\right)\right\rceil,
    \end{equation}
    and the signal processing algorithm in \Cref{subsec:spnogap} is used for spectrum estimation in \Cref{alg:main}, then with probability at least $1-\p$, the output $E_L$ satisfies
    \begin{equation}
        \La\subset E_L\subset B(\La, \eps).
    \end{equation}
    The maximum runtime and total runtime are 
    \begin{equation}
        \tm = \mo{\eta K \eps^{-1}}, \quad \tt = \mo{\eta K^2 N_\HR\eps^{-1}}.
    \end{equation}
\end{theorem}

\begin{proof}
    For the $N_\HR$ defined above, one knows from Hoeffding's inequality that for any $\ell$ and $k$, 
    \begin{equation}
	\PP\left(\abs{{y}_\ell(k) - \hat{f}_\ell(k)} < \alpha\right) > 1-\frac{\p}{(\left\lceil\log_2\frac{\eta}{\eps}\right\rceil+1)(K+1)}.
  \end{equation}
  Thus $\abs{{y}_\ell(k) - \hat{f}_\ell(k)} < \alpha$ is true for all $\ell$ and $k$ with probability at least $1-\p$ by the union bound, and the rest follows from \Cref{lem:inclusion} and \Cref{lemma:m_real}.
\end{proof}

\begin{cor}\label{cor: real complexity}
    By choosing $\eta= \frac{1}{8S(2S-1)}$ and $K=3\eta^{-1}\tau$, the maximum runtime is $\tm = \mo{\log(\frac{1}{\beta-\omega}){\eps}^{-1}}$ and the total complexity is $\tt = \tmo{(\beta-\omega)^{-2}S^2{\eps}^{-1}}$, which achieves the Heisenberg limit. 
\end{cor}

\subsection{The integer-power model}\label{sec:int}
This subsection shows that if $U$ is given as a black box and only (positive) integer powers of $U$ can be accessed, \Cref{alg:main} can be applied with the help of prime numbers. As for the initialization, we take $E_{-1}=[0,1]$. Similar with \Cref{sec:int}, we can verify the uniqueness of $q_{\ell, i}$ in \Cref{alg:main} with help of the following lemma:

\begin{lemma}\label{lem:inclusion2}
    Given integer $M_{\ell-1}$ and the set $E_{\ell-1}$, if the choice of integer $m_\ell $ satisfies 
    \begin{equation}
    \begin{aligned}
    \bigg(B_\TT\left(E_{\ell-1},\frac{\eta}{2M_{\ell-1}}\right) + &\frac{q}{M_{\ell-1}m_\ell }\bigg)\cap E_{\ell-1} = \varnothing \mod 1\\
    &\text{for all $q\in\ZZ$ that $(m_\ell M_{\ell-1})\nmid q$},
    \end{aligned}
    \end{equation} 
    then the choice of $q_{\ell ,i}$ in \Cref{alg:main} is unique, and the corresponding $E_\ell $ satisfies
    \begin{equation}\label{lem:inclusionint}
        E_\ell  \subset B_\TT\left(\La,\frac{\eta}{2M_{\ell-1}}\right)\subset B_\TT\left(E_{\ell-1},\frac{\eta}{2M_{\ell-1}}\right).
    \end{equation}
\end{lemma}

This lemma is the modulo-$1$ version of \Cref{lem:inclusion}, which can be directly obtained from \Cref{lem:inclusion} since $\{m_\ell\}$ are all integers and all the sets we consider here are in $\RR/\ZZ$. 
In what follows, we denote the $i$-th prime number by $p_i$. Here we define $p_0=1$ to unify the notations. The factor $m_\ell$ in the algorithm above can be chosen with the help of the following lemma. The proof is provided in \Cref{sec:primedistproof}
\begin{lemma}\label{lemma:primedist}
    For any $(\theta_1, \ldots, \theta_t)\in\RR^t$ and $\frac{t(t-1)}{2}+1$ prime numbers $p_1<p_2<\cdots p_{\frac{t(t-1)}{2}+1}$, there is at least one $p_l\in\{p_1, p_2, \ldots, p_{\frac{t(t-1)}{2}+1}\}$ such that
    \begin{equation}\label{eq:primeineq}
       \underset{\substack{1\leq i\le j\leq t\\p_l\nmid k}}{\min}~ \left||\theta_i-\theta_j|-\frac{k}{p_l}\right|\geq\frac{1}{2p_{\frac{t(t-1)}{2}}p_{\frac{t(t-1)}{2}+1}}.
    \end{equation}
\end{lemma}
\begin{rmk}
    The result also applies to mutually prime integers that are not necessarily prime themselves. Moreover, this procedure only involves classical computing with polynomial complexity in terms of $t$, so it can be implemented efficiently. 
\end{rmk}

Based on \Cref{lemma:primedist}, one can show that there is at least one $m_\ell \in \{p_1, p_2, \ldots, p_{\frac{S(S-1)}{2}+1}\}$ that satisfies the requirements for \Cref{alg:main} if $\eta<\left(3Sp_{\frac{S(S-1)}{2}}p_{\frac{S(S-1)}{2}+1}\right)^{-1}$. Since $p_n = \mo{n\log n}$ (\cite{hadamard1896distribution,poussin1897recherches}),
the requirement for $\eta$ is $\eta = \mo{S^{-5}\log^{-2}(S)}$. More precisely, in \cite{rosser1962approximate} it was proved that $p_n<n(\log n + 2\log\log n)$ for $n\ge 4$, thus by direct calculation one knows $p_n\le 2n(\log n+1)$ for any $n\ge 1$. Hence, it suffices to have $\eta\le \min\{\frac{1}{6}, \frac{1}{3S^5(0.31+2\log S)^2}\}$. 

To establish \eqref{lem:inclusionint}, we also need the following lemma, which is proved in \Cref{sec:intmkproof}
\begin{lemma}\label{thm:intmk}
  Suppose $\eta<\left(3Sp_{\frac{S(S-1)}{2}}p_{\frac{S(S-1)}{2}+1}\right)^{-1}$, then for step $\ell$ ($\ell \geq1$) in \Cref{alg:main}, one can choose a prime number $m_\ell $ such that 
  \begin{equation}\label{eq:stepmkpruned}
  \begin{aligned}
    &\left(B_\TT(E_{\ell-1}, \frac{\eta}{M_{\ell-1}})+\frac{q}{m_\ell M_{\ell-1}}\right)\\
    &\cap B_\TT(E_{\ell-1}, \frac{\eta}{M_{\ell-1}}) = \varnothing\mod 1, ~\text{if $(m_\ell M_{\ell-1})\nmid q$},
  \end{aligned}
  \end{equation}
  which guarantees the construction of $E_\ell $ in the algorithm. 
\end{lemma}

Similar to \Cref{thm:complexityreal}, one directly obtains the following theorem by applying Hoeffding's inequality and the union bound to \Cref{lem:inclusion2} and \Cref{thm:intmk}. 
\begin{theorem}\label{thm:complexityint}
    Define $\tau = \frac{1}{\pi}\log\frac{12}{\beta-\omega}$ and let $\alpha$ be a constant that satisfies $\alpha<\frac{\beta-\omega}{3}$. For any $\eta\le \min\{\frac{1}{6}, \frac{1}{3S^5(0.31+2\log S)^2}\}$, if $K \ge 3\eta^{-1}\tau$ and
    \begin{equation}\label{eq:NHR}
       N_\HR = 2\left\lceil\frac{4}{\alpha^2}\left(\log\frac{4}{\p}+\log\left(\left\lceil\log_2\frac{\eta}{\eps}\right\rceil+1\right)+\log(K+1)\right)\right\rceil,
    \end{equation}
    and the signal processing algorithm in \Cref{subsec:spnogap} is used for spectrum estimation in \Cref{alg:main}, then with probability at least $1-\p$, the output $E_L$ satisfies
    \begin{equation}
        \La\subset E_L\subset B_\TT(\La, \eps).
    \end{equation}
    The maximum runtime and total runtime are 
    \begin{equation}
        \tm = \mo{\eta K \eps^{-1}}, \quad \tt = \mo{\eta K^2 N_\HR\eps^{-1}}.
    \end{equation}
\end{theorem}

\begin{cor}
    By choosing $\eta= \min\{\frac{1}{6}, \frac{1}{3S^5(0.31+2\log S)^2}\}$ and $K=3\eta^{-1}\tau$, the maximum runtime is $\tm = \tmo{\log(\frac{1}{\beta-\omega}){\eps}^{-1}}$ and the total complexity is $\tt = \tmo{(\beta-\omega)^{-2}S^5{\eps}^{-1}}$, which achieves the Heisenberg limit. 
\end{cor}

\section{The gapped case} \label{sec:sp}

This section specializes \Cref{alg:main} to the simpler gapped case, i.e., there is a minimum separation among the dominant eigenvalues defined as:
\begin{equation}\label{eq:gap}
  \Delta = \min_{1\le s\le s'\le S, ~n\in \ZZ} |\lambda_s-\lambda_{s'}-n|.
\end{equation}
We first review the famous signal processing routine ESPRIT in \Cref{sec:spgap}. In particular, the number of frequencies $K$ needed only depends on $S$ and $\Delta$. With the help of this specific version of ESPRIT, we show that \Cref{alg:main} meets the requirements \ref{en:req1}, \ref{en:req2}, \ref{en:req3} and \ref{en:req4}. In particular, the maximum runtime $\tm$ can scale like $\omega/\beta$. The real-power and integer-power models are investigated in \Cref{sec:nonint2} and \Cref{sec:int2}, respectively.
A graphical illustration of the algorithm for the gapless case is given in Figure \ref{fig:gapped}.

\begin{figure*}[ht]
  \centering
  \includegraphics[scale=0.4]{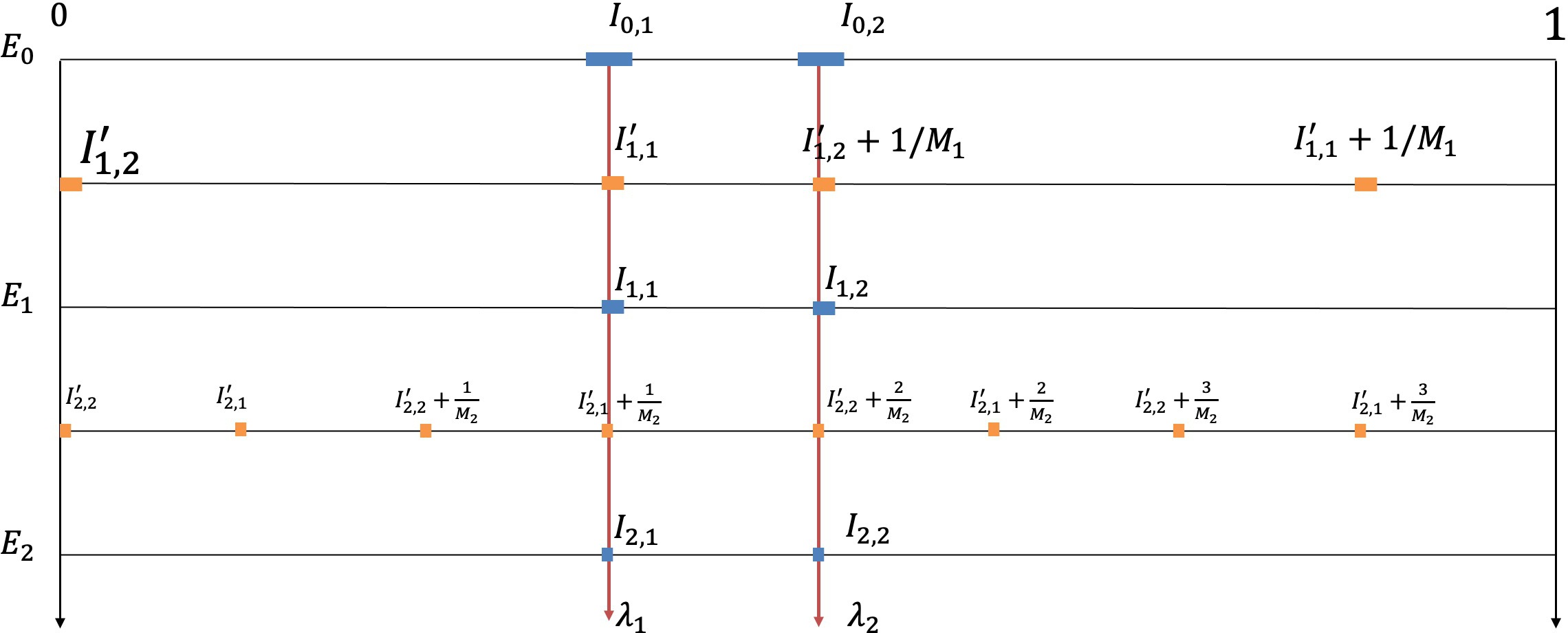} 
  \caption{Illustration of \Cref{alg:main} in the gapped case with $S=2$. Based on prior knowledge of the separation between $\lambda_1$ and $\lambda_2$, one can obtain much tighter initial-stage estimations than the gapless case (see \Cref{fig:gapless}) if the residual $\omega$ is small. As a result, much fewer steps are needed to reach the desired precision $\eps$, and the maximum runtime is reduced. }
  \label{fig:gapped}
\end{figure*}

\subsection{Signal processing routine}\label{sec:spgap}

Since a finite-size gap between the dominant eigenvalues is available, multiple signal processing algorithms can be used in this setting, such as the ones in \cite{morgenshtern2016super, denoyelle2015support, li2020super}. Here, we adopt a particular version of the ESPRIT algorithm discussed in \cite{li2020super}. 

Suppose that $K \ge 4S$ is even. Recall that at the $\ell$-th step, the data ${y}_\ell(k)$ collected from the Hadamard test results satisfy $|{y}_\ell(k)-\hat{f}_\ell(k)|\le \alpha$, where $\hat{f}_\ell$ is defined in \eqref{eq:fhatk}. The first step in ESPRIT is to construct the following Hankel matrix.
\[
    H_\ell=\begin{bmatrix}
{y}_\ell(0) & {y}_\ell(1) & \ldots & {y}_\ell({\frac{K}{2}}) \\
{y}_\ell(1) & {y}_\ell(2) & \ldots & {y}_\ell({\frac{K}{2}+1}) \\
\vdots & \vdots & \ddots & \vdots \\
{y}_\ell({\frac{K}{2}}) & {y}_\ell({\frac{K}{2}+1}) & \cdots & {y}_\ell(K)
\end{bmatrix}
\]
Then, one applies the singular value decomposition (SVD) to $H_\ell$ and obtains
\[
H_\ell = [U_\ell, U_\ell^\perp] \Sigma [V_\ell, V_\ell^\perp]^*,
\]
where $U_\ell$ has $S$ columns. Let the first $K/2$ rows of $U_\ell$ be $U_\ell^{(0)}$ and let its last $K/2$ rows be $U_\ell^{(1)}$. The last step is to compute the eigenvalues $\{\mu_1, \ldots, \mu_S\}$ of the matrix $(U_\ell^{(0)})^\dagger U_\ell^{(1)}$, where $(U_\ell^{(0)})^\dagger$ is the Moore-Penrose pseudo-inverse of $U_\ell^{(0)}$. The output is the set 
\[
\tilde{\Lambda}_\ell:=\left\{-\frac{1}{2\pi}\arg(\mu_1), \ldots, -\frac{1}{2\pi}\arg(\mu_S)\right\}.
\] 

The following theorem is proved in \cite{li2020super} for this particular version of ESPRIT:
\begin{theorem}\label{thm:esprit}
    Suppose $K \ge 4S$ is even and $K \ge \frac{4}{\tilde{\Delta}}$, where $\tilde{\Delta}\ge\Delta_\ell:= \min_{1\le i<j\le S}{|M_\ell\lambda_i-M_\ell\lambda_j|_s}$. For any constant $C\in(2, \frac{K \tilde{\Delta}}{2})$, if 
    \begin{equation}\label{eq:espritcond}
        \omega+\alpha \le \frac{C\beta}{8(C-1)\sqrt{2S}}\sqrt{1-\frac{2(C-1)S}{CK}},
    \end{equation}
    then
    \begin{equation}
        \mathrm{md}(\tilde{\Lambda}_\ell, \La_\ell) \le \frac{40S^2}{\beta}\sqrt{\frac{C^3(2+K)}{(C-1)^3 K}}\left(1-\frac{2CS}{(C-1)K}\right)^{-1} (\omega+\alpha),
    \end{equation}
    where $\mathrm{md}(\cdot,\cdot)$ denotes the matching distance between two finite sets with the same cardinality. 
\end{theorem}

We emphasize that the approximation error is controlled by the noise in the gapped case. The number of frequencies $K$ only needs to be proportional to $S$, which is considered a fixed number. This means that the maximum depth of the circuit can be kept fixed. Only more repetitions are needed to bring down the noise level. 

When $\omega$ is sufficiently small compared with $\beta$, \eqref{eq:espritcond} can always be satisfied by properly choosing $\alpha$. The set $Y_\ell$ needed in the main algorithm can then be defined as
\[
Y_\ell = B_\TT(\tilde{\Lambda}_\ell, \frac{\eta}{2}),
\]
with $\eta$ satisfying
\[
\frac{80S^2}{\beta}\sqrt{\frac{C^3(2+K)}{(C-1)^3K}}\left(1-\frac{2CS}{(C-1)K}\right)^{-1} (\omega+\alpha)<\eta.
\]
This $Y_\ell$ is guaranteed to satisfy the three properties stated in \Cref{thm:yk}.

In the following sections, we provide modified versions of the algorithms described in \Cref{sec:ns} and show that when $\omega$ is close to $0$, one can further reduce the maximal runtime by a factor $\Omega(\omega/\beta)$.

\subsection{The real-power model}\label{sec:nonint2}
The following theorem states the $\tm$ and $\tt$ bounds for our RMPE algorithm for the real-power model under the gapped case. The proof can be found in \Cref{sec:complexityreal 2proof}

\begin{theorem}\label{thm:complexityreal 2}
    Let $\tilde{\Delta}<\min\{\frac{1}{8S(2S-1)}, \Delta\}$, $K>4/\tilde{\Delta}$ and $\eta$ be any number such that
    \begin{equation}\label{eq:cond eta}
    \begin{aligned}
    \frac{80S^2}{\beta}\sqrt{\frac{C^3(2+K)}{(C-1)^3K}}&\left(1-\frac{2CS}{(C-1)K}\right)^{-1} \omega  \\
    &<\eta <\frac{1}{8S(2S-1)},
    \end{aligned}
    \end{equation}
    where $C$ is a constant such that $C\in(2, \frac{K\tilde{\Delta}}{2})$.
    Suppose that $\alpha$ satisfy 
    \begin{equation}\label{eq:cond alpha}
    \begin{aligned}
     &\alpha<\min\Bigg\{\frac{C\beta}{8(C-1)\sqrt{2S}}\sqrt{1-\frac{2(C-1)S}{CK}}, \\ 
     &\frac{\beta}{80S^2}\sqrt{\frac{(C-1)^3K}{C^3(2+K)}}\left(1-\frac{2CS}{(C-1)K}\right)\eta\Bigg\}  - \omega,
    \end{aligned}
    \end{equation}
    and $N_\HR$ satisfy
    \begin{equation}
       N_\HR = 2\left\lceil\frac{4}{\alpha^2}\left(\log\frac{4}{\p}+\log\left(\left\lceil\log_2\frac{\eta}{\eps}\right\rceil+1\right)+\log(K+1)\right)\right\rceil,
    \end{equation}
    and the ESPRIT routine described in \Cref{sec:spgap} is used for spectrum estimation in \Cref{alg:main}, then with probability at least $1-\p$, the output $E_L$ satisfies
    \begin{equation}
        \La\subset E_L\subset B(\La, \eps).
    \end{equation}
    The maximum runtime and total runtime are 
    \begin{equation}
        \tm = \mo{\eta K \eps^{-1}}, \quad \tt = \mo{\eta K^ 2N_\HR\eps^{-1}}.
    \end{equation}
\end{theorem}

When the left inequality of \eqref{eq:cond eta} is tight, i.e., $\eta = \Theta(\frac{\omega}{\beta}S^2)$, we will have the following corollary that gives a smaller prefactor of $\tm$ compared with \Cref{cor: real complexity}.
\begin{cor}\label{cor: real complexity prefactor}
    By setting $\eta = \Theta(\frac{\omega}{\beta}S^2)$, the complexity bounds given in the theorem are
    \begin{equation}
    \begin{aligned}
    &\tm = \mo{(\Delta^{-1}+S^2)S^2\omega \beta^{-1}\eps^{-1}}, \\ 
    &\tt = \tmo{(\Delta^{-1}+S^2)S^2\omega^{-1}\beta^{-1}\eps^{-1}}.
    \end{aligned}
    \end{equation}
    The prefactor of $\tm$ is proportional to $\omega\beta^{-1}$, the ratio between the energy of the residual in $\ket{\psi}$ and the energy lower bound of the dominant eigenmodes.
\end{cor}
Therefore, the maximum circuit depth prefactor can be made small when this ratio is sufficiently small. However, $\tt$ will increase for small $\omega$, so there is a trade-off between maximal runtime and total runtime, and one should choose $\eta$ and $\omega$ appropriately in practice. 

It is worth noticing that $\omega$ is not fixed when the initial state $\ket{\psi}$ is given. Instead, we can adjust $\omega$ as long as $\omega\ge|c_\res|^2$ and \eqref{eq:cond eta} is feasible.

\subsection{The integer-power model}\label{sec:int2}

Similar to \Cref{sec:int}, one can implement the algorithm with only integer powers of $U$. The following theorem states the $\tm$ and $\tt$ bounds for our RMPE algorithm for the integer-power model under the gapped case. We give the proof of it in \Cref{sec:thm:complexityint 2proof}.

\begin{theorem}\label{thm:complexityint 2}
    Let $\tilde{\Delta}<\min\{\frac{1}{2p_{\half{S(S-1)}}p_{\half{S(S-1)}+1}}, \Delta\}$, $K>4/\tilde{\Delta}$ and $\eta$ be any number such that
    \[
    \begin{aligned}
    \frac{80S^2}{\beta}&\sqrt{\frac{C^3(2+K)}{(C-1)^3K}}\left(1-\frac{2CS}{(C-1)K}\right)^{-1} \omega \\ 
    &< \eta <\frac{1}{4Sp_{\frac{S(S-1)}{2}}p_{\frac{S(S-1)}{2}+1}},
    \end{aligned}
    \]
    where $C$ is a constant such that $C\in(2, \frac{K \Delta}{2})$.
    Let $\alpha$ satisfy 
    \[
    \begin{aligned}
     \alpha&<\min\Bigg\{\frac{C\beta}{8(C-1)\sqrt{2S}}\sqrt{1-\frac{2(C-1)S}{CK}}, \\ 
     &\frac{\beta}{80S^2}\sqrt{\frac{(C-1)^3K}{C^3(2+K)}}\left(1-\frac{2CS}{(C-1)K}\right)\eta\Bigg\}  - \omega,
    \end{aligned}
    \]
    and
    \begin{equation}
       N_\HR = 2\left\lceil\frac{4}{\alpha^2}\left(\log\frac{4}{\p}+\log\left(\left\lceil\log_2\frac{\eta}{\eps}\right\rceil+1\right)+\log(K+1)\right)\right\rceil.
    \end{equation}
    Suppose ESPRIT is used for spectrum estimation in \Cref{alg:main} and $m_\ell$ is chosen according to \Cref{thm:intmk}, then with probability at least $1-\p$, the output $E_L$ satisfies
    \begin{equation}
        \La\subset E_L\subset B(\La, \eps).
    \end{equation}
    The maximum runtime and total runtime are 
    \begin{equation}
        \tm = \mo{\eta K \eps^{-1}}, \quad \tt = \mo{\eta K^2 N_\HR\eps^{-1}}.
    \end{equation}
\end{theorem}

\begin{cor}\label{cor:prefacsp}
    When $\omega$ is sufficiently close to zero, one can take $\eta = \Theta(\frac{\omega}{\beta}S^2)$ and  $K =\Theta(\frac{4}{\tilde{\Delta}})$, then we have
    \begin{equation}
    \begin{aligned}
    &\tm = \tmo{(\Delta^{-1}+S^4)S^2\omega \beta^{-1}\eps^{-1}}, \\ 
    &\tt = \tmo{(\Delta^{-1}+S^4)S^2\omega^{-1}\beta^{-1}\eps^{-1}},
    \end{aligned}
    \end{equation}
    Here, the maximum circuit depth prefactor again scales like $\mo{\omega\beta^{-1}}$. 
\end{cor}

\section{Hybrid algorithm with improved prefactor}\label{sec:hybrid}

Assuming a finite spectral gap $\Delta$, the results in \Cref{sec:sp} give a prefactor $\omega\beta^{-1}$ in $\tm$ (see \Cref{cor: real complexity prefactor} and \Cref{cor:prefacsp}). However, when the gap $\Delta$ is a small, the prefactor $\Delta^{-1}$ in $\tm$ is undesirable.  On the other hand, the results in \Cref{sec:ns} work for an arbitrarily small spectral gap but cannot provide a prefactor as in \Cref{cor: real complexity prefactor} and \Cref{cor:prefacsp} due to the signal processing technique used.

In this section, we propose to combine the methods in \Cref{sec:ns} and \Cref{sec:sp} under the general framework of \Cref{alg:main}. Intuitively, if the gap $\Delta$ is small, one first applies the signal processing technique in \Cref{subsec:spnogap} for certain iterations as a burn-in period and then switches to the ESPRIT routine in \Cref{sec:spgap}. We prove that this hybrid method provides an improved prefactor in $\tm$.

\subsection{The real-power model}\label{sec:noint3}
The following theorem gives the detailed description and complexity bound of the hybrid algorithm when we have access to the real powers of the given unitary $U$, which is proved in \Cref{sec:real hybridproof}.

\begin{theorem}\label{thm:real hybrid}
    Let $\tilde{\Delta}$ be any number in $(0,\frac{1}{8S(2S-1)})$,  $\tilde{M}:=\tilde{\Delta}/\Delta$, integer $K_2>4/\tilde{\Delta}$ and $\eta$ be any number such that
    \begin{equation}\label{eq:cond eta3}
    \begin{aligned}
    \frac{80S^2}{\beta}\sqrt{\frac{C^3(2+K_2)}{(C-1)^3K_2}}&\left(1-\frac{2CS}{(C-1)K_2}\right)^{-1} \omega \\ 
    &< \eta <\frac{1}{8S(2S-1)},
    \end{aligned}
    \end{equation}
    where $C$ is a constant such that $C\in(2, \frac{K_2\tilde{\Delta}}{2})$.
    
    In \Cref{alg:main}, when $M_\ell\le \tilde{M}$, we use the signal processing subroutine in \Cref{sec:ns} with parameters $\tau = \frac{1}{\pi}\log\frac{12}{\beta-\omega}$, $\alpha_1=\frac{\beta-\omega}{3}$, $K_1= \lceil 3\eta^{-1}\tau\rceil$    and
    \begin{equation}
       N_{\HR,1} = 2\left\lceil\frac{4}{\alpha_1^2}\left(\log\frac{4}{\p}+\log\left(\left\lceil\log_2\frac{\eta}{\eps}\right\rceil+1\right)+\log(K_1+1)\right)\right\rceil.
    \end{equation} 

    When $M_\ell > \tilde{M}$, we use ESPRIT as the signal processing subroutine with parameters $\tilde{\Delta},$ $\alpha_2$, $K_2$, $N_{\HR,2}$ where
    \begin{equation}\label{eq:cond alpha3}
    \begin{aligned}
     \alpha_2&<\min\Bigg\{\frac{C\beta}{8(C-1)\sqrt{2S}}\sqrt{1-\frac{2(C-1)S}{CK_2}}, \\ 
     &\frac{\beta}{80S^2}\sqrt{\frac{(C-1)^3K_2}{C^3(2+K_2)}}\left(1-\frac{2CS}{(C-1)K_2}\right)\eta\Bigg\}  - \omega,
    \end{aligned}
    \end{equation}
    \begin{equation}
       N_{\HR,2} = 2\left\lceil\frac{4}{\alpha_2^2}\left(\log\frac{4}{\p}+\log\left(\left\lceil\log_2\frac{\eta}{\eps}\right\rceil+1\right)+\log(K_2+1)\right)\right\rceil.
    \end{equation}

    Under this setting, the maximal runtime of the algorithm is
    \begin{equation}
        \tm = \mo{\max\left\{\log\left(\frac{1}{\beta}\right)\eta^{-1}\tilde{\Delta}\Delta^{-1}, \eta\tilde{\Delta}^{-1}\eps^{-1}\right\}}
    \end{equation}
    and the total runtime is
    \begin{equation}
        \tt = \tmo{\eta^{-2}\beta^{-2}\tilde{\Delta}\Delta^{-1}+ \eta\omega^{-2}\tilde{\Delta}^{-2}\eps^{-1}}.
    \end{equation}   
\end{theorem}

\begin{cor}
    If we choose $\tilde{\Delta} = \frac{1}{8S(2S-1)}$ and $\eta = \Theta(\omega\beta^{-1}S^2)$ in \Cref{thm:real hybrid}, then the maximal and total runtime of the algorithm are
    \begin{equation}
    \begin{aligned}
    &\tm = O_{S,\beta}\left(\max\{\omega^{-1}\Delta^{-1}, \omega\eps^{-1}\}\right), \\ 
    &\tt = \tilde{O}_{S,\beta}\left(\omega^{-2}\Delta^{-1}+ \omega^{-1}\eps^{-1}\right).
    \end{aligned}
    \end{equation}
\end{cor}
In the case that $S$ and $\frac{1}{\beta}$ are moderate values that can be treated as constants, this scaling $O$ is better than the one in \Cref{cor: real complexity prefactor}. This improves the complexity from $\Delta^{-1}\eps^{-1}$ to $\Delta^{-1}+\eps^{-1}$ while retaining the prefactor $\omega$ in front of $\eps^{-1}$ in $\tm$.

\subsection{The integer-power model}\label{sec:int3}
Now, we give the analysis of the hybrid algorithm for the integer-power model. The proof for the following result can be found in \Cref{sec:int hybridproof}
\begin{theorem}\label{thm:int hybrid}
  Let $\tdm\ge2$ such that $\tilde{\Delta}:= \tdm \Delta < \frac{1}{4Sp_{\frac{S(S-1)}{2}}p_{\frac{S(S-1)}{2}+1}}$. Let $\eta\le \frac{1}{6}$ such that
    \begin{equation}\label{eq:cond eta4}
    \begin{aligned}
      \frac{80S^2}{\beta}&\sqrt{\frac{C^3(2+K_2)}{(C-1)^3K_2}}\left(1-\frac{2CS}{(C-1)K_2}\right)^{-1} \omega \\ 
      &< \eta <\frac{1}{4Sp_{\frac{S(S-1)}{2}}p_{\frac{S(S-1)}{2}+1}},
    \end{aligned}
    \end{equation}
    where $K_2$ is an integer that satisfies $K_2>4/\tilde{\Delta}$ and $C$ is a constant such that $C\in(2, \frac{K_2\tilde{\Delta}}{2})$.
    
    Under the framework of \Cref{alg:main}, the hybrid algorithm uses the signal processing subroutine in \Cref{subsec:spnogap} with parameters $\tau = \frac{1}{\pi}\log\frac{12}{\beta-\omega}$, $\alpha_1=\frac{\beta-\omega}{3}$, $K_1= \lceil 3\eta^{-1}\tau\rceil$    and
    \begin{equation}
       N_{\HR,1} = 2\left\lceil\frac{4}{\alpha_1^2}\left(\log\frac{4}{\p}+\log\left(\left\lceil\log_2\frac{\eta}{\eps}\right\rceil+1\right)+\log(K_1+1)\right)\right\rceil,
    \end{equation} 
    when $M_\ell\le \tdm$. 

    When $M_\ell > \tilde{M}$, the hybrid algorithm uses ESPRIT (see \Cref{sec:spgap}) as the signal processing subroutine with parameters $\tilde{\Delta},$ $\alpha_2$, $K_2$, $N_{\HR,2}$ where
    \begin{equation}\label{eq:cond alpha4}
    \begin{aligned}
     \alpha_2&<\min\Bigg\{\frac{C\beta}{8(C-1)\sqrt{2S}}\sqrt{1-\frac{2(C-1)S}{CK_2}},\\
     &\frac{\beta}{80S^2}\sqrt{\frac{(C-1)^3K_2}{C^3(2+K_2)}}\left(1-\frac{2CS}{(C-1)K_2}\right)\eta\Bigg\}  - \omega,
    \end{aligned}
    \end{equation}
    and
    \begin{equation}
       N_{\HR,2} = 2\left\lceil\frac{4}{\alpha_2^2}\left(\log\frac{4}{\p}+\log\left(\left\lceil\log_2\frac{\eta}{\eps}\right\rceil+1\right)+\log(K_2+1)\right)\right\rceil.
    \end{equation}
    Then with probability at least $1-\p$, the output $E_L$ satisfies
    \begin{equation}
        \La\subset E_L\subset B(\La, \eps).
    \end{equation}
    The maximal runtime of the algorithm is
    \begin{equation}
        \tm = \tilde{O}_S\left(\max\{\eta^{-1}\tilde{\Delta}\Delta^{-1}, \eta\tilde{\Delta}^{-1}\eps^{-1}\}\right)
    \end{equation}
    and the total runtime is
    \begin{equation}
        \tt =  \tilde{O}_S\left(\eta^{-2}\tilde{\Delta}\Delta^{-1}N_{HR,1}+\eta\tilde{\Delta}^{-2}\eps^{-1}N_{HR,2}\right).
    \end{equation}   
\end{theorem}

\begin{cor}\label{cor:hybridint}
    If we choose $\tilde{\Delta} = \frac{1}{4Sp_{\frac{S(S-1)}{2}}p_{\frac{S(S-1)}{2}+1}}$ and $\eta = \Theta(\omega\beta^{-1}S^2)$ in \Cref{thm:int hybrid}, then the maximal and total runtime of the algorithm are
    \begin{equation}
    \begin{aligned}
        &\tm = O_{S,\beta}\left(\max\{\omega^{-1}\Delta^{-1}, \omega\eps^{-1}\}\right), \\ 
        &\tt = \tilde{O}_{S,\beta}\left(\omega^{-2}\Delta^{-1}+ \omega^{-1}\eps^{-1}\right).        
    \end{aligned}
    \end{equation}
\end{cor}
The asymptotic scaling of \Cref{cor:hybridint} is better than the one obtained in \Cref{cor:prefacsp}: the maximum runtime is improved from $\Delta^{-1}\eps^{-1}$ to $\Delta^{-1}+\eps^{-1}$, while the $\eps^{-1}$ term still enjoys the prefactor $\omega$.

\section{Discussions}

This paper proposed robust multiple-phase estimation (RMPE) algorithms for estimating multiple eigenvalues. The proposed algorithms address both the gapless case in \Cref{sec:ns} and the gapped case in \Cref{sec:sp}, \ref{sec:hybrid} and for both the integer-power and real-power models. 

When the problem is gapless, according to \Cref{sec:ns}, the number of frequencies $K$ needs to be increased if one needs an estimation with a smaller error level $\eta$ from the signal processing routine, which prevents us from reducing the maximum runtime even when the magnitude of the residual $\omega$ is close to zero. When a finite-size spectral gap is assumed, the number of frequencies $K$ is a constant independent of $\eta$. An immediate direction is to explore signal processing algorithms that can improve $K$'s dependence on $\eta$ and the spectral gap. This enables us to reduce the prefactor in the maximum runtime when $\omega$ is small.

Another relevant problem comes from the bound in \Cref{lemma:primedist}:
\[
\max_{1\le l\le \frac{t(t-1)}{2}+1} \min_{\substack{1\leq i\le j\leq t\\p_l\nmid k}} \left||\theta_i-\theta_j|-\frac{k}{p_l}\right|\geq\frac{1}{2p_{\frac{t(t-1)}{2}}p_{\frac{t(t-1)}{2}+1}},
\]
for any $\{\omega\}_{i=1}^{t}$. However, one can also show that
\begin{equation}
    \underset{(\theta_i)_{i=1}^t}{\inf}~\underset{1\le l\le t'}{\max}~\underset{\substack{1\leq i\le j\leq t\\p_l\nmid k}}{\min}~ \left||\theta_i-\theta_j|-\frac{k}{p_l}\right|=\gamma_t(t')>0,
\end{equation}
for any $t'>t$. The conclusion of \Cref{lemma:primedist} is thus $\gamma_t(\half t(t-1)+1)\ge \frac{1}{2p_{\frac{t(t-1)}{2}}p_{\frac{t(t-1)}{2}+1}}$. It would be interesting to find a sharper estimation for $\gamma_t(t')$ and, thereby, the optimal choice of $t'$ in terms of the overall complexity.

\begin{acknowledgments}
We thank Wenjing Liao for helpful discussions on line spectral estimation. We thank Zhiyan Ding and Lin Lin for explaining the work in \cite{ding2023sim} and providing comments and suggestions. The work of L.Y. is partially supported by the National Science Foundation under awards DMS-2011699 and DMS-2208163.
\end{acknowledgments}

H.L. and H.N. contributed equally to this work.

\nocite{*}

\bibliography{apssamp}

\providecommand{\noopsort}[1]{}\providecommand{\singleletter}[1]{#1}%
\begin{thebibliography}{39}%
\makeatletter
\providecommand \@ifxundefined [1]{%
 \@ifx{#1\undefined}
}%
\providecommand \@ifnum [1]{%
 \ifnum #1\expandafter \@firstoftwo
 \else \expandafter \@secondoftwo
 \fi
}%
\providecommand \@ifx [1]{%
 \ifx #1\expandafter \@firstoftwo
 \else \expandafter \@secondoftwo
 \fi
}%
\providecommand \natexlab [1]{#1}%
\providecommand \enquote  [1]{``#1''}%
\providecommand \bibnamefont  [1]{#1}%
\providecommand \bibfnamefont [1]{#1}%
\providecommand \citenamefont [1]{#1}%
\providecommand \href@noop [0]{\@secondoftwo}%
\providecommand \href [0]{\begingroup \@sanitize@url \@href}%
\providecommand \@href[1]{\@@startlink{#1}\@@href}%
\providecommand \@@href[1]{\endgroup#1\@@endlink}%
\providecommand \@sanitize@url [0]{\catcode `\\12\catcode `\$12\catcode `\&12\catcode `\#12\catcode `\^12\catcode `\_12\catcode `\%12\relax}%
\providecommand \@@startlink[1]{}%
\providecommand \@@endlink[0]{}%
\providecommand \url  [0]{\begingroup\@sanitize@url \@url }%
\providecommand \@url [1]{\endgroup\@href {#1}{\urlprefix }}%
\providecommand \urlprefix  [0]{URL }%
\providecommand \Eprint [0]{\href }%
\providecommand \doibase [0]{https://doi.org/}%
\providecommand \selectlanguage [0]{\@gobble}%
\providecommand \bibinfo  [0]{\@secondoftwo}%
\providecommand \bibfield  [0]{\@secondoftwo}%
\providecommand \translation [1]{[#1]}%
\providecommand \BibitemOpen [0]{}%
\providecommand \bibitemStop [0]{}%
\providecommand \bibitemNoStop [0]{.\EOS\space}%
\providecommand \EOS [0]{\spacefactor3000\relax}%
\providecommand \BibitemShut  [1]{\csname bibitem#1\endcsname}%
\let\auto@bib@innerbib\@empty
\bibitem [{\citenamefont {Kimmel}\ \emph {et~al.}(2015)\citenamefont {Kimmel}, \citenamefont {Low},\ and\ \citenamefont {Yoder}}]{kimmel2015robust}%
  \BibitemOpen
  \bibfield  {author} {\bibinfo {author} {\bibfnamefont {S.}~\bibnamefont {Kimmel}}, \bibinfo {author} {\bibfnamefont {G.~H.}\ \bibnamefont {Low}},\ and\ \bibinfo {author} {\bibfnamefont {T.~J.}\ \bibnamefont {Yoder}},\ }\bibfield  {title} {\bibinfo {title} {Robust calibration of a universal single-qubit gate set via robust phase estimation},\ }\href@noop {} {\bibfield  {journal} {\bibinfo  {journal} {Physical Review A}\ }\textbf {\bibinfo {volume} {92}},\ \bibinfo {pages} {062315} (\bibinfo {year} {2015})}\BibitemShut {NoStop}%
\bibitem [{\citenamefont {Belliardo}\ and\ \citenamefont {Giovannetti}(2020)}]{belliardo2020achieving}%
  \BibitemOpen
  \bibfield  {author} {\bibinfo {author} {\bibfnamefont {F.}~\bibnamefont {Belliardo}}\ and\ \bibinfo {author} {\bibfnamefont {V.}~\bibnamefont {Giovannetti}},\ }\bibfield  {title} {\bibinfo {title} {Achieving {H}eisenberg scaling with maximally entangled states: An analytic upper bound for the attainable root-mean-square error},\ }\href@noop {} {\bibfield  {journal} {\bibinfo  {journal} {Physical Review A}\ }\textbf {\bibinfo {volume} {102}},\ \bibinfo {pages} {042613} (\bibinfo {year} {2020})}\BibitemShut {NoStop}%
\bibitem [{\citenamefont {Russo}\ \emph {et~al.}(2021)\citenamefont {Russo}, \citenamefont {Rudinger}, \citenamefont {Morrison},\ and\ \citenamefont {Baczewski}}]{russo2021evaluating}%
  \BibitemOpen
  \bibfield  {author} {\bibinfo {author} {\bibfnamefont {A.~E.}\ \bibnamefont {Russo}}, \bibinfo {author} {\bibfnamefont {K.~M.}\ \bibnamefont {Rudinger}}, \bibinfo {author} {\bibfnamefont {B.~C.}\ \bibnamefont {Morrison}},\ and\ \bibinfo {author} {\bibfnamefont {A.~D.}\ \bibnamefont {Baczewski}},\ }\bibfield  {title} {\bibinfo {title} {Evaluating energy differences on a quantum computer with robust phase estimation},\ }\href@noop {} {\bibfield  {journal} {\bibinfo  {journal} {Physical review letters}\ }\textbf {\bibinfo {volume} {126}},\ \bibinfo {pages} {210501} (\bibinfo {year} {2021})}\BibitemShut {NoStop}%
\bibitem [{\citenamefont {Li}\ \emph {et~al.}(2023)\citenamefont {Li}, \citenamefont {Ni},\ and\ \citenamefont {Ying}}]{li2023note}%
  \BibitemOpen
  \bibfield  {author} {\bibinfo {author} {\bibfnamefont {H.}~\bibnamefont {Li}}, \bibinfo {author} {\bibfnamefont {H.}~\bibnamefont {Ni}},\ and\ \bibinfo {author} {\bibfnamefont {L.}~\bibnamefont {Ying}},\ }\href {https://arxiv.org/abs/2303.00946} {\bibinfo {title} {A note on spike localization for line spectrum estimation}} (\bibinfo {year} {2023})\BibitemShut {NoStop}%
\bibitem [{\citenamefont {Giovannetti}\ \emph {et~al.}(2006)\citenamefont {Giovannetti}, \citenamefont {Lloyd},\ and\ \citenamefont {Maccone}}]{giovannetti2006quantum}%
  \BibitemOpen
  \bibfield  {author} {\bibinfo {author} {\bibfnamefont {V.}~\bibnamefont {Giovannetti}}, \bibinfo {author} {\bibfnamefont {S.}~\bibnamefont {Lloyd}},\ and\ \bibinfo {author} {\bibfnamefont {L.}~\bibnamefont {Maccone}},\ }\bibfield  {title} {\bibinfo {title} {Quantum metrology},\ }\href@noop {} {\bibfield  {journal} {\bibinfo  {journal} {Physical review letters}\ }\textbf {\bibinfo {volume} {96}},\ \bibinfo {pages} {010401} (\bibinfo {year} {2006})}\BibitemShut {NoStop}%
\bibitem [{\citenamefont {Zwierz}\ \emph {et~al.}(2010)\citenamefont {Zwierz}, \citenamefont {P{\'e}rez-Delgado},\ and\ \citenamefont {Kok}}]{zwierz2010general}%
  \BibitemOpen
  \bibfield  {author} {\bibinfo {author} {\bibfnamefont {M.}~\bibnamefont {Zwierz}}, \bibinfo {author} {\bibfnamefont {C.~A.}\ \bibnamefont {P{\'e}rez-Delgado}},\ and\ \bibinfo {author} {\bibfnamefont {P.}~\bibnamefont {Kok}},\ }\bibfield  {title} {\bibinfo {title} {General optimality of the {H}eisenberg limit for quantum metrology},\ }\href@noop {} {\bibfield  {journal} {\bibinfo  {journal} {Physical review letters}\ }\textbf {\bibinfo {volume} {105}},\ \bibinfo {pages} {180402} (\bibinfo {year} {2010})}\BibitemShut {NoStop}%
\bibitem [{\citenamefont {Zhou}\ \emph {et~al.}(2018)\citenamefont {Zhou}, \citenamefont {Zhang}, \citenamefont {Preskill},\ and\ \citenamefont {Jiang}}]{zhou2018achieving}%
  \BibitemOpen
  \bibfield  {author} {\bibinfo {author} {\bibfnamefont {S.}~\bibnamefont {Zhou}}, \bibinfo {author} {\bibfnamefont {M.}~\bibnamefont {Zhang}}, \bibinfo {author} {\bibfnamefont {J.}~\bibnamefont {Preskill}},\ and\ \bibinfo {author} {\bibfnamefont {L.}~\bibnamefont {Jiang}},\ }\bibfield  {title} {\bibinfo {title} {Achieving the {H}eisenberg limit in quantum metrology using quantum error correction},\ }\href@noop {} {\bibfield  {journal} {\bibinfo  {journal} {Nature communications}\ }\textbf {\bibinfo {volume} {9}},\ \bibinfo {pages} {78} (\bibinfo {year} {2018})}\BibitemShut {NoStop}%
\bibitem [{\citenamefont {Kitaev}(1995)}]{kitaev1995quantum}%
  \BibitemOpen
  \bibfield  {author} {\bibinfo {author} {\bibfnamefont {A.~Y.}\ \bibnamefont {Kitaev}},\ }\bibfield  {title} {\bibinfo {title} {Quantum measurements and the abelian stabilizer problem},\ }\href@noop {} {\bibfield  {journal} {\bibinfo  {journal} {arXiv preprint quant-ph/9511026}\ } (\bibinfo {year} {1995})}\BibitemShut {NoStop}%
\bibitem [{\citenamefont {Kitaev}\ \emph {et~al.}(2002)\citenamefont {Kitaev}, \citenamefont {Shen}, \citenamefont {Vyalyi},\ and\ \citenamefont {Vyalyi}}]{kitaev2002classical}%
  \BibitemOpen
  \bibfield  {author} {\bibinfo {author} {\bibfnamefont {A.~Y.}\ \bibnamefont {Kitaev}}, \bibinfo {author} {\bibfnamefont {A.}~\bibnamefont {Shen}}, \bibinfo {author} {\bibfnamefont {M.~N.}\ \bibnamefont {Vyalyi}},\ and\ \bibinfo {author} {\bibfnamefont {M.~N.}\ \bibnamefont {Vyalyi}},\ }\href@noop {} {\emph {\bibinfo {title} {Classical and quantum computation}}},\ \bibinfo {number} {47}\ (\bibinfo  {publisher} {American Mathematical Soc.},\ \bibinfo {year} {2002})\BibitemShut {NoStop}%
\bibitem [{\citenamefont {Cleve}\ \emph {et~al.}(1998)\citenamefont {Cleve}, \citenamefont {Ekert}, \citenamefont {Macchiavello},\ and\ \citenamefont {Mosca}}]{cleve1998quantum}%
  \BibitemOpen
  \bibfield  {author} {\bibinfo {author} {\bibfnamefont {R.}~\bibnamefont {Cleve}}, \bibinfo {author} {\bibfnamefont {A.}~\bibnamefont {Ekert}}, \bibinfo {author} {\bibfnamefont {C.}~\bibnamefont {Macchiavello}},\ and\ \bibinfo {author} {\bibfnamefont {M.}~\bibnamefont {Mosca}},\ }\bibfield  {title} {\bibinfo {title} {Quantum algorithms revisited},\ }\href@noop {} {\bibfield  {journal} {\bibinfo  {journal} {Proceedings of the Royal Society of London. Series A: Mathematical, Physical and Engineering Sciences}\ }\textbf {\bibinfo {volume} {454}},\ \bibinfo {pages} {339} (\bibinfo {year} {1998})}\BibitemShut {NoStop}%
\bibitem [{\citenamefont {Berry}\ \emph {et~al.}(2015)\citenamefont {Berry}, \citenamefont {Childs}, \citenamefont {Cleve}, \citenamefont {Kothari},\ and\ \citenamefont {Somma}}]{berry2015simulating}%
  \BibitemOpen
  \bibfield  {author} {\bibinfo {author} {\bibfnamefont {D.~W.}\ \bibnamefont {Berry}}, \bibinfo {author} {\bibfnamefont {A.~M.}\ \bibnamefont {Childs}}, \bibinfo {author} {\bibfnamefont {R.}~\bibnamefont {Cleve}}, \bibinfo {author} {\bibfnamefont {R.}~\bibnamefont {Kothari}},\ and\ \bibinfo {author} {\bibfnamefont {R.~D.}\ \bibnamefont {Somma}},\ }\bibfield  {title} {\bibinfo {title} {Simulating {H}amiltonian dynamics with a truncated {T}aylor series},\ }\href@noop {} {\bibfield  {journal} {\bibinfo  {journal} {Physical review letters}\ }\textbf {\bibinfo {volume} {114}},\ \bibinfo {pages} {090502} (\bibinfo {year} {2015})}\BibitemShut {NoStop}%
\bibitem [{\citenamefont {Higgins}\ \emph {et~al.}(2007)\citenamefont {Higgins}, \citenamefont {Berry}, \citenamefont {Bartlett}, \citenamefont {Wiseman},\ and\ \citenamefont {Pryde}}]{higgins2007entanglement}%
  \BibitemOpen
  \bibfield  {author} {\bibinfo {author} {\bibfnamefont {B.~L.}\ \bibnamefont {Higgins}}, \bibinfo {author} {\bibfnamefont {D.~W.}\ \bibnamefont {Berry}}, \bibinfo {author} {\bibfnamefont {S.~D.}\ \bibnamefont {Bartlett}}, \bibinfo {author} {\bibfnamefont {H.~M.}\ \bibnamefont {Wiseman}},\ and\ \bibinfo {author} {\bibfnamefont {G.~J.}\ \bibnamefont {Pryde}},\ }\bibfield  {title} {\bibinfo {title} {Entanglement-free {H}eisenberg-limited phase estimation},\ }\href@noop {} {\bibfield  {journal} {\bibinfo  {journal} {Nature}\ }\textbf {\bibinfo {volume} {450}},\ \bibinfo {pages} {393} (\bibinfo {year} {2007})}\BibitemShut {NoStop}%
\bibitem [{\citenamefont {Knill}\ \emph {et~al.}(2007)\citenamefont {Knill}, \citenamefont {Ortiz},\ and\ \citenamefont {Somma}}]{knill2007optimal}%
  \BibitemOpen
  \bibfield  {author} {\bibinfo {author} {\bibfnamefont {E.}~\bibnamefont {Knill}}, \bibinfo {author} {\bibfnamefont {G.}~\bibnamefont {Ortiz}},\ and\ \bibinfo {author} {\bibfnamefont {R.~D.}\ \bibnamefont {Somma}},\ }\bibfield  {title} {\bibinfo {title} {Optimal quantum measurements of expectation values of observables},\ }\href@noop {} {\bibfield  {journal} {\bibinfo  {journal} {Physical Review A}\ }\textbf {\bibinfo {volume} {75}},\ \bibinfo {pages} {012328} (\bibinfo {year} {2007})}\BibitemShut {NoStop}%
\bibitem [{\citenamefont {Poulin}\ and\ \citenamefont {Wocjan}(2009)}]{poulin2009sampling}%
  \BibitemOpen
  \bibfield  {author} {\bibinfo {author} {\bibfnamefont {D.}~\bibnamefont {Poulin}}\ and\ \bibinfo {author} {\bibfnamefont {P.}~\bibnamefont {Wocjan}},\ }\bibfield  {title} {\bibinfo {title} {Sampling from the thermal quantum gibbs state and evaluating partition functions with a quantum computer},\ }\href@noop {} {\bibfield  {journal} {\bibinfo  {journal} {Physical review letters}\ }\textbf {\bibinfo {volume} {103}},\ \bibinfo {pages} {220502} (\bibinfo {year} {2009})}\BibitemShut {NoStop}%
\bibitem [{\citenamefont {O’Brien}\ \emph {et~al.}(2019)\citenamefont {O’Brien}, \citenamefont {Tarasinski},\ and\ \citenamefont {Terhal}}]{o2019quantum}%
  \BibitemOpen
  \bibfield  {author} {\bibinfo {author} {\bibfnamefont {T.~E.}\ \bibnamefont {O’Brien}}, \bibinfo {author} {\bibfnamefont {B.}~\bibnamefont {Tarasinski}},\ and\ \bibinfo {author} {\bibfnamefont {B.~M.}\ \bibnamefont {Terhal}},\ }\bibfield  {title} {\bibinfo {title} {Quantum phase estimation of multiple eigenvalues for small-scale (noisy) experiments},\ }\href@noop {} {\bibfield  {journal} {\bibinfo  {journal} {New Journal of Physics}\ }\textbf {\bibinfo {volume} {21}},\ \bibinfo {pages} {023022} (\bibinfo {year} {2019})}\BibitemShut {NoStop}%
\bibitem [{\citenamefont {Dong}\ \emph {et~al.}(2022)\citenamefont {Dong}, \citenamefont {Lin},\ and\ \citenamefont {Tong}}]{dong2022ground}%
  \BibitemOpen
  \bibfield  {author} {\bibinfo {author} {\bibfnamefont {Y.}~\bibnamefont {Dong}}, \bibinfo {author} {\bibfnamefont {L.}~\bibnamefont {Lin}},\ and\ \bibinfo {author} {\bibfnamefont {Y.}~\bibnamefont {Tong}},\ }\bibfield  {title} {\bibinfo {title} {Ground-state preparation and energy estimation on early fault-tolerant quantum computers via quantum eigenvalue transformation of unitary matrices},\ }\href@noop {} {\bibfield  {journal} {\bibinfo  {journal} {PRX Quantum}\ }\textbf {\bibinfo {volume} {3}},\ \bibinfo {pages} {040305} (\bibinfo {year} {2022})}\BibitemShut {NoStop}%
\bibitem [{\citenamefont {Lin}\ and\ \citenamefont {Tong}(2020)}]{lin2020near}%
  \BibitemOpen
  \bibfield  {author} {\bibinfo {author} {\bibfnamefont {L.}~\bibnamefont {Lin}}\ and\ \bibinfo {author} {\bibfnamefont {Y.}~\bibnamefont {Tong}},\ }\bibfield  {title} {\bibinfo {title} {Near-optimal ground state preparation},\ }\href@noop {} {\bibfield  {journal} {\bibinfo  {journal} {Quantum}\ }\textbf {\bibinfo {volume} {4}},\ \bibinfo {pages} {372} (\bibinfo {year} {2020})}\BibitemShut {NoStop}%
\bibitem [{\citenamefont {Lin}\ and\ \citenamefont {Tong}(2022)}]{lin2022heisenberg}%
  \BibitemOpen
  \bibfield  {author} {\bibinfo {author} {\bibfnamefont {L.}~\bibnamefont {Lin}}\ and\ \bibinfo {author} {\bibfnamefont {Y.}~\bibnamefont {Tong}},\ }\bibfield  {title} {\bibinfo {title} {{H}eisenberg-limited ground-state energy estimation for early fault-tolerant quantum computers},\ }\href@noop {} {\bibfield  {journal} {\bibinfo  {journal} {PRX Quantum}\ }\textbf {\bibinfo {volume} {3}},\ \bibinfo {pages} {010318} (\bibinfo {year} {2022})}\BibitemShut {NoStop}%
\bibitem [{\citenamefont {Ding}\ and\ \citenamefont {Lin}(2022)}]{ding2022even}%
  \BibitemOpen
  \bibfield  {author} {\bibinfo {author} {\bibfnamefont {Z.}~\bibnamefont {Ding}}\ and\ \bibinfo {author} {\bibfnamefont {L.}~\bibnamefont {Lin}},\ }\bibfield  {title} {\bibinfo {title} {Even shorter quantum circuit for phase estimation on early fault-tolerant quantum computers with applications to ground-state energy estimation},\ }\href@noop {} {\bibfield  {journal} {\bibinfo  {journal} {arXiv preprint arXiv:2211.11973}\ } (\bibinfo {year} {2022})}\BibitemShut {NoStop}%
\bibitem [{\citenamefont {Nielsen}\ and\ \citenamefont {Chuang}(2001)}]{nielsen2001quantum}%
  \BibitemOpen
  \bibfield  {author} {\bibinfo {author} {\bibfnamefont {M.~A.}\ \bibnamefont {Nielsen}}\ and\ \bibinfo {author} {\bibfnamefont {I.~L.}\ \bibnamefont {Chuang}},\ }\bibfield  {title} {\bibinfo {title} {Quantum computation and quantum information},\ }\href@noop {} {\bibfield  {journal} {\bibinfo  {journal} {Phys. Today}\ }\textbf {\bibinfo {volume} {54}},\ \bibinfo {pages} {60} (\bibinfo {year} {2001})}\BibitemShut {NoStop}%
\bibitem [{\citenamefont {Ni}\ \emph {et~al.}(2023)\citenamefont {Ni}, \citenamefont {Li},\ and\ \citenamefont {Ying}}]{ni2023kitaev}%
  \BibitemOpen
  \bibfield  {author} {\bibinfo {author} {\bibfnamefont {H.}~\bibnamefont {Ni}}, \bibinfo {author} {\bibfnamefont {H.}~\bibnamefont {Li}},\ and\ \bibinfo {author} {\bibfnamefont {L.}~\bibnamefont {Ying}},\ }\href {https://arxiv.org/abs/2302.02454} {\bibinfo {title} {On low-depth algorithms for quantum phase estimation}} (\bibinfo {year} {2023})\BibitemShut {NoStop}%
\bibitem [{\citenamefont {Somma}(2019)}]{somma2019quantum}%
  \BibitemOpen
  \bibfield  {author} {\bibinfo {author} {\bibfnamefont {R.~D.}\ \bibnamefont {Somma}},\ }\bibfield  {title} {\bibinfo {title} {Quantum eigenvalue estimation via time series analysis},\ }\href@noop {} {\bibfield  {journal} {\bibinfo  {journal} {New Journal of Physics}\ }\textbf {\bibinfo {volume} {21}},\ \bibinfo {pages} {123025} (\bibinfo {year} {2019})}\BibitemShut {NoStop}%
\bibitem [{\citenamefont {Dutkiewicz}\ \emph {et~al.}(2022)\citenamefont {Dutkiewicz}, \citenamefont {Terhal},\ and\ \citenamefont {O'Brien}}]{dutkiewicz2022heisenberg}%
  \BibitemOpen
  \bibfield  {author} {\bibinfo {author} {\bibfnamefont {A.}~\bibnamefont {Dutkiewicz}}, \bibinfo {author} {\bibfnamefont {B.~M.}\ \bibnamefont {Terhal}},\ and\ \bibinfo {author} {\bibfnamefont {T.~E.}\ \bibnamefont {O'Brien}},\ }\bibfield  {title} {\bibinfo {title} {Heisenberg-limited quantum phase estimation of multiple eigenvalues with few control qubits},\ }\href@noop {} {\bibfield  {journal} {\bibinfo  {journal} {Quantum}\ }\textbf {\bibinfo {volume} {6}},\ \bibinfo {pages} {830} (\bibinfo {year} {2022})}\BibitemShut {NoStop}%
\bibitem [{\citenamefont {Ding}\ and\ \citenamefont {Lin}(2023)}]{ding2023sim}%
  \BibitemOpen
  \bibfield  {author} {\bibinfo {author} {\bibfnamefont {Z.}~\bibnamefont {Ding}}\ and\ \bibinfo {author} {\bibfnamefont {L.}~\bibnamefont {Lin}},\ }\bibfield  {title} {\bibinfo {title} {Simultaneous estimation of multiple eigenvalues with short-depth quantum circuit on early fault-tolerant quantum computers},\ }\href@noop {} {\bibfield  {journal} {\bibinfo  {journal} {arXiv:2303.05714}\ } (\bibinfo {year} {2023})}\BibitemShut {NoStop}%
\bibitem [{\citenamefont {McClean}\ \emph {et~al.}(2017)\citenamefont {McClean}, \citenamefont {Kimchi-Schwartz}, \citenamefont {Carter},\ and\ \citenamefont {de~Jong}}]{mcclean2017hybrid}%
  \BibitemOpen
  \bibfield  {author} {\bibinfo {author} {\bibfnamefont {J.~R.}\ \bibnamefont {McClean}}, \bibinfo {author} {\bibfnamefont {M.~E.}\ \bibnamefont {Kimchi-Schwartz}}, \bibinfo {author} {\bibfnamefont {J.}~\bibnamefont {Carter}},\ and\ \bibinfo {author} {\bibfnamefont {W.~A.}\ \bibnamefont {de~Jong}},\ }\bibfield  {title} {\bibinfo {title} {Hybrid quantum-classical hierarchy for mitigation of decoherence and determination of excited states},\ }\href@noop {} {\bibfield  {journal} {\bibinfo  {journal} {Phys. Rev. A}\ }\textbf {\bibinfo {volume} {95}},\ \bibinfo {pages} {042308} (\bibinfo {year} {2017})}\BibitemShut {NoStop}%
\bibitem [{\citenamefont {Motta}\ \emph {et~al.}(2020)\citenamefont {Motta}, \citenamefont {Sun}, \citenamefont {Tan}, \citenamefont {O’Rourke}, \citenamefont {Ye}, \citenamefont {Minnich}, \citenamefont {Brand{\~a}o},\ and\ \citenamefont {Chan}}]{motta2020determining}%
  \BibitemOpen
  \bibfield  {author} {\bibinfo {author} {\bibfnamefont {M.}~\bibnamefont {Motta}}, \bibinfo {author} {\bibfnamefont {C.}~\bibnamefont {Sun}}, \bibinfo {author} {\bibfnamefont {A.~T.}\ \bibnamefont {Tan}}, \bibinfo {author} {\bibfnamefont {M.~J.}\ \bibnamefont {O’Rourke}}, \bibinfo {author} {\bibfnamefont {E.}~\bibnamefont {Ye}}, \bibinfo {author} {\bibfnamefont {A.~J.}\ \bibnamefont {Minnich}}, \bibinfo {author} {\bibfnamefont {F.~G.}\ \bibnamefont {Brand{\~a}o}},\ and\ \bibinfo {author} {\bibfnamefont {G.~K.-L.}\ \bibnamefont {Chan}},\ }\bibfield  {title} {\bibinfo {title} {Determining eigenstates and thermal states on a quantum computer using quantum imaginary time evolution},\ }\href@noop {} {\bibfield  {journal} {\bibinfo  {journal} {Nature Physics}\ }\textbf {\bibinfo {volume} {16}},\ \bibinfo {pages} {205} (\bibinfo {year} {2020})}\BibitemShut {NoStop}%
\bibitem [{\citenamefont {Cortes}\ and\ \citenamefont {Gray}(2022)}]{cortes2022quantum}%
  \BibitemOpen
  \bibfield  {author} {\bibinfo {author} {\bibfnamefont {C.~L.}\ \bibnamefont {Cortes}}\ and\ \bibinfo {author} {\bibfnamefont {S.~K.}\ \bibnamefont {Gray}},\ }\bibfield  {title} {\bibinfo {title} {Quantum krylov subspace algorithms for ground-and excited-state energy estimation},\ }\href@noop {} {\bibfield  {journal} {\bibinfo  {journal} {Physical Review A}\ }\textbf {\bibinfo {volume} {105}},\ \bibinfo {pages} {022417} (\bibinfo {year} {2022})}\BibitemShut {NoStop}%
\bibitem [{\citenamefont {Klymko}\ \emph {et~al.}(2022)\citenamefont {Klymko}, \citenamefont {Mejuto-Zaera}, \citenamefont {Cotton}, \citenamefont {Wudarski}, \citenamefont {Urbanek}, \citenamefont {Hait}, \citenamefont {Head-Gordon}, \citenamefont {Whaley}, \citenamefont {Moussa}, \citenamefont {Wiebe}, \citenamefont {de~Jong},\ and\ \citenamefont {Tubman}}]{klymko2022real}%
  \BibitemOpen
  \bibfield  {author} {\bibinfo {author} {\bibfnamefont {K.}~\bibnamefont {Klymko}}, \bibinfo {author} {\bibfnamefont {C.}~\bibnamefont {Mejuto-Zaera}}, \bibinfo {author} {\bibfnamefont {S.~J.}\ \bibnamefont {Cotton}}, \bibinfo {author} {\bibfnamefont {F.}~\bibnamefont {Wudarski}}, \bibinfo {author} {\bibfnamefont {M.}~\bibnamefont {Urbanek}}, \bibinfo {author} {\bibfnamefont {D.}~\bibnamefont {Hait}}, \bibinfo {author} {\bibfnamefont {M.}~\bibnamefont {Head-Gordon}}, \bibinfo {author} {\bibfnamefont {K.~B.}\ \bibnamefont {Whaley}}, \bibinfo {author} {\bibfnamefont {J.}~\bibnamefont {Moussa}}, \bibinfo {author} {\bibfnamefont {N.}~\bibnamefont {Wiebe}}, \bibinfo {author} {\bibfnamefont {W.~A.}\ \bibnamefont {de~Jong}},\ and\ \bibinfo {author} {\bibfnamefont {N.~M.}\ \bibnamefont {Tubman}},\ }\bibfield  {title} {\bibinfo {title} {Real-time evolution for ultracompact hamiltonian eigenstates on quantum hardware},\ }\href {https://doi.org/10.1103/PRXQuantum.3.020323} {\bibfield  {journal} {\bibinfo  {journal}
  {PRX Quantum}\ }\textbf {\bibinfo {volume} {3}},\ \bibinfo {pages} {020323} (\bibinfo {year} {2022})}\BibitemShut {NoStop}%
\bibitem [{\citenamefont {Parrish}\ and\ \citenamefont {McMahon}(2019)}]{parrish2019quantum}%
  \BibitemOpen
  \bibfield  {author} {\bibinfo {author} {\bibfnamefont {R.~M.}\ \bibnamefont {Parrish}}\ and\ \bibinfo {author} {\bibfnamefont {P.~L.}\ \bibnamefont {McMahon}},\ }\bibfield  {title} {\bibinfo {title} {Quantum filter diagonalization: Quantum eigendecomposition without full quantum phase estimation},\ }\href@noop {} {\bibfield  {journal} {\bibinfo  {journal} {arXiv preprint arXiv:1909.08925}\ } (\bibinfo {year} {2019})}\BibitemShut {NoStop}%
\bibitem [{\citenamefont {Seki}\ and\ \citenamefont {Yunoki}(2021)}]{seki2021quantum}%
  \BibitemOpen
  \bibfield  {author} {\bibinfo {author} {\bibfnamefont {K.}~\bibnamefont {Seki}}\ and\ \bibinfo {author} {\bibfnamefont {S.}~\bibnamefont {Yunoki}},\ }\bibfield  {title} {\bibinfo {title} {Quantum power method by a superposition of time-evolved states},\ }\href@noop {} {\bibfield  {journal} {\bibinfo  {journal} {PRX Quantum}\ }\textbf {\bibinfo {volume} {2}},\ \bibinfo {pages} {010333} (\bibinfo {year} {2021})}\BibitemShut {NoStop}%
\bibitem [{\citenamefont {Epperly}\ \emph {et~al.}(2022)\citenamefont {Epperly}, \citenamefont {Lin},\ and\ \citenamefont {Nakatsukasa}}]{epperly2022theory}%
  \BibitemOpen
  \bibfield  {author} {\bibinfo {author} {\bibfnamefont {E.~N.}\ \bibnamefont {Epperly}}, \bibinfo {author} {\bibfnamefont {L.}~\bibnamefont {Lin}},\ and\ \bibinfo {author} {\bibfnamefont {Y.}~\bibnamefont {Nakatsukasa}},\ }\bibfield  {title} {\bibinfo {title} {A theory of quantum subspace diagonalization},\ }\href@noop {} {\bibfield  {journal} {\bibinfo  {journal} {SIAM Journal on Matrix Analysis and Applications}\ }\textbf {\bibinfo {volume} {43}},\ \bibinfo {pages} {1263} (\bibinfo {year} {2022})}\BibitemShut {NoStop}%
\bibitem [{\citenamefont {Roy}\ and\ \citenamefont {Kailath}(1989)}]{roy1989esprit}%
  \BibitemOpen
  \bibfield  {author} {\bibinfo {author} {\bibfnamefont {R.}~\bibnamefont {Roy}}\ and\ \bibinfo {author} {\bibfnamefont {T.}~\bibnamefont {Kailath}},\ }\bibfield  {title} {\bibinfo {title} {Esprit-estimation of signal parameters via rotational invariance techniques},\ }\href@noop {} {\bibfield  {journal} {\bibinfo  {journal} {IEEE Transactions on acoustics, speech, and signal processing}\ }\textbf {\bibinfo {volume} {37}},\ \bibinfo {pages} {984} (\bibinfo {year} {1989})}\BibitemShut {NoStop}%
\bibitem [{\citenamefont {Hadamard}(1896)}]{hadamard1896distribution}%
  \BibitemOpen
  \bibfield  {author} {\bibinfo {author} {\bibfnamefont {J.}~\bibnamefont {Hadamard}},\ }\bibfield  {title} {\bibinfo {title} {Sur la distribution des z{\'e}ros de la fonction $\zeta(s)$ et ses cons{\'e}quences arithm{\'e}tiques},\ }\href@noop {} {\bibfield  {journal} {\bibinfo  {journal} {Bulletin de la Societ{\'e} mathematique de France}\ }\textbf {\bibinfo {volume} {24}},\ \bibinfo {pages} {199} (\bibinfo {year} {1896})}\BibitemShut {NoStop}%
\bibitem [{\citenamefont {Poussin}(1897)}]{poussin1897recherches}%
  \BibitemOpen
  \bibfield  {author} {\bibinfo {author} {\bibfnamefont {C.~J. d. L.~V.}\ \bibnamefont {Poussin}},\ }\href@noop {} {\emph {\bibinfo {title} {Recherches analytiques sur la th{\'e}orie des nombres premiers}}},\ Vol.~\bibinfo {volume} {1}\ (\bibinfo  {publisher} {Hayez},\ \bibinfo {year} {1897})\BibitemShut {NoStop}%
\bibitem [{\citenamefont {Rosser}\ and\ \citenamefont {Schoenfeld}(1962)}]{rosser1962approximate}%
  \BibitemOpen
  \bibfield  {author} {\bibinfo {author} {\bibfnamefont {J.~B.}\ \bibnamefont {Rosser}}\ and\ \bibinfo {author} {\bibfnamefont {L.}~\bibnamefont {Schoenfeld}},\ }\bibfield  {title} {\bibinfo {title} {Approximate formulas for some functions of prime numbers},\ }\href@noop {} {\bibfield  {journal} {\bibinfo  {journal} {Illinois Journal of Mathematics}\ }\textbf {\bibinfo {volume} {6}},\ \bibinfo {pages} {64} (\bibinfo {year} {1962})}\BibitemShut {NoStop}%
\bibitem [{\citenamefont {Morgenshtern}\ and\ \citenamefont {Candes}(2016)}]{morgenshtern2016super}%
  \BibitemOpen
  \bibfield  {author} {\bibinfo {author} {\bibfnamefont {V.~I.}\ \bibnamefont {Morgenshtern}}\ and\ \bibinfo {author} {\bibfnamefont {E.~J.}\ \bibnamefont {Candes}},\ }\bibfield  {title} {\bibinfo {title} {Super-resolution of positive sources: The discrete setup},\ }\href@noop {} {\bibfield  {journal} {\bibinfo  {journal} {SIAM Journal on Imaging Sciences}\ }\textbf {\bibinfo {volume} {9}},\ \bibinfo {pages} {412} (\bibinfo {year} {2016})}\BibitemShut {NoStop}%
\bibitem [{\citenamefont {Denoyelle}\ \emph {et~al.}(2015)\citenamefont {Denoyelle}, \citenamefont {Duval},\ and\ \citenamefont {Peyr{\'e}}}]{denoyelle2015support}%
  \BibitemOpen
  \bibfield  {author} {\bibinfo {author} {\bibfnamefont {Q.}~\bibnamefont {Denoyelle}}, \bibinfo {author} {\bibfnamefont {V.}~\bibnamefont {Duval}},\ and\ \bibinfo {author} {\bibfnamefont {G.}~\bibnamefont {Peyr{\'e}}},\ }\bibfield  {title} {\bibinfo {title} {Support recovery for sparse deconvolution of positive measures},\ }\href@noop {} {\bibfield  {journal} {\bibinfo  {journal} {arXiv preprint arXiv:1506.08264}\ } (\bibinfo {year} {2015})}\BibitemShut {NoStop}%
\bibitem [{\citenamefont {Li}\ \emph {et~al.}(2020)\citenamefont {Li}, \citenamefont {Liao},\ and\ \citenamefont {Fannjiang}}]{li2020super}%
  \BibitemOpen
  \bibfield  {author} {\bibinfo {author} {\bibfnamefont {W.}~\bibnamefont {Li}}, \bibinfo {author} {\bibfnamefont {W.}~\bibnamefont {Liao}},\ and\ \bibinfo {author} {\bibfnamefont {A.}~\bibnamefont {Fannjiang}},\ }\bibfield  {title} {\bibinfo {title} {Super-resolution limit of the esprit algorithm},\ }\href@noop {} {\bibfield  {journal} {\bibinfo  {journal} {IEEE transactions on information theory}\ }\textbf {\bibinfo {volume} {66}},\ \bibinfo {pages} {4593} (\bibinfo {year} {2020})}\BibitemShut {NoStop}%
\bibitem [{\citenamefont {Birell}\ and\ \citenamefont {Davies}(1982)}]{Bire82}%
  \BibitemOpen
  \bibfield  {author} {\bibinfo {author} {\bibfnamefont {N.~D.}\ \bibnamefont {Birell}}\ and\ \bibinfo {author} {\bibfnamefont {P.~C.~W.}\ \bibnamefont {Davies}},\ }\href@noop {} {\emph {\bibinfo {title} {Quantum Fields in Curved Space}}}\ (\bibinfo  {publisher} {Cambridge University Press},\ \bibinfo {year} {1982})\BibitemShut {NoStop}%
\end{thebibliography}%
\clearpage
\onecolumngrid
\appendix
\section{Proofs}\label{sec:append}
\subsection{Proof of Corollary~\ref{thm:yk}}
\label{sec:ykproof}
\begin{proof}
    All the sets in the proof are in a modulo-1 sense because $X_\ell$ can be viewed as a subset of $\RR/\ZZ$. Since $X_\ell$ is the level set of a continuous function, it can be written as the disjoint union of finitely many intervals
    \[X_\ell = \bigcup_{i=1}^{q}(a_i,b_i),\]
    where $0\le a_1\le b_1\le a_2\le \cdots\le a_q\le b_q\le 1$.
    Let $J = \{1\le j\le q: a_j-b_{j-1}<\tau/K\}$, where $a_1-b_{0}<\tau/K$ means $a_1+1-b_q<\tau/K$ since the intervals are in modulo-1 sense.
    Then we may define
    \[Y_\ell = \left(\bigcup_{i=1}^{q}[a_i,b_i]\right)\cup\left(\bigcup_{j\in J}[b_{j-1},a_j]\right).\]
    In this way, we have $\La_\ell\subset X_\ell \subset Y_\ell$, and $Y_\ell \subset B_\TT(X_\ell,\frac{\tau}{2L})\subset B_\TT(\La_\ell,\frac{3\tau}{2L})\subset B_\TT(\La_\ell,\frac{\eta}{2})$, which means $\abs{Y_\ell}\le 3S\tau/K \le S\eta$.

    By definition, $Y_\ell$ is the disjoint union of some intervals of the form $[a_{i_n},b_{j_n}]$, where $a_{i_n}-b_{j_{n-1}}\ge a_{i_n}-b_{i_n-1}\ge \tau/K$ and similarly $a_{i_{n+1}}-b_{j_n}\ge \tau/K$. If it does not contain any spike, then $\half(a_{i_n}+b_{i_n})$ will be at least $\tau/K$ away from any spike. However, we know $\half(a_{i_n}+b_{i_n})\in X_\ell$, which contradicts with \Cref{thm:spnogap}.
\end{proof}

\subsection{Proof of Lemma~\ref{lemma:m_real}}
\label{sec:m_realproof}
\begin{proof}
    For fixed $i$ and $j$, we define the following set:
    \begin{equation}
    \begin{aligned}
        R_{ij} := [2,4]\cap\bigg\{m: \exists\ q\in\ZZ\backslash\{0\}\text{ such that }\abs{\theta_i-\theta_j-\frac{q}{Mm}}\le \frac{1}{M}(\zeta_i+\zeta_j)\bigg\},\label{eq:uij}
    \end{aligned}
    \end{equation}
    which is the set of all $m$'s that violate \eqref{eq:disjoint}.

    First, consider the case $i=j$. From the bound of $\zeta$, we can deduce $\frac{1}{Mm} \ge \frac{1}{4M}\ge\frac{2S\zeta}{M}>\frac{1}{M}(2\zeta_i)$, and thus $R_{ii}=\varnothing$ for every $i$.

    For the case $i\ne j$, we can assume $\theta_i>\theta_j$ without loss of generality because clearly $R_{ij} = R_{ji}$. Notice that the existence of such $i$ and $j$ implies $t\geq2$. In this case, the bound of $\zeta$ again gives $\theta_i-\theta_j+\frac{1}{Mm}>\frac{1}{Mm}\ge\frac{1}{4M}\ge\frac{S\zeta}{M}>\frac{1}{M}(\zeta_i+\zeta_j)$, which means that only positive $q$'s can violate \eqref{eq:disjoint}. Therefore, we can rewrite \eqref{eq:uij} as
    \begin{equation}
    \begin{aligned}
        &R_{ij} = [2,4]\cap\left(\bigcup_{q=1}^{\infty}\left[\frac{q}{M\xiij+\deij},\frac{q}{M\xiij-\deij}\right]\right).
    \end{aligned}
    \end{equation}
    We consider two cases to estimate $|R_{ij}|$. In the case that $M\xiij\le\frac{1}{4}-\deij$, we have $\frac{q}{M\xiij+\deij}\ge4$ for every $q$, so $|R_{ij}| = 0$. Another case is that $M\xiij>\frac{1}{4}-\deij$, then the maximal $q$ that $\frac{q}{M\xiij+\deij}\le 4$ is $q_{\max}\le 4(M\xiij+\deij)$. Therefore
    \begin{equation}
        \begin{aligned}
            |R_{ij}|&
            \le \sum_{q=1}^{q_{\max}}\abs{\frac{q}{M\xiij-\deij} - \frac{q}{M\xiij+\deij}}\\
            &=\frac{2\deij\sum_{q=1}^{q_{\max}}q}{\left(M\xiij-\deij\right)\left(M\xiij+\deij\right)}\\
            &\le \frac{2\deij\left(4(M\xiij+\deij)\right)^2}{\left(M\xiij-\deij\right)\left(M\xiij+\deij\right)}\\
            &= 32\deij\left(1+\frac{2\deij}{M\xiij-\deij}\right)\\
            &< 32\deij\left(1+\frac{2\deij}{\frac{1}{4}-2\deij}\right)\\
            &= \frac{32\deij}{\left(1-8\deij\right)}.\\
        \end{aligned}
    \end{equation}
    Taking all pairs $(i,j)$ into account, we have
    \begin{equation}
      \begin{aligned}
        \abs{\bigcup_{1\le i,j\le t} R_{ij}} &= \abs{\bigcup_{1\le i<j\le t} R_{ij}}\le \sum_{1\le i<j\le t}\frac{32\deij}{\left(1-8\deij\right)}\\
        &< \sum_{1\le i<j\le t}\frac{32\deij}{\left(1-8S\zeta\right)}
        = \frac{32(t-1)}{1-8S\zeta}\sum_{i=1}^t \zeta_i        \\
        &\le \frac{32(S-1)S\zeta}{1-8S\zeta}\le 2 = 4-2,
        \end{aligned}
    \end{equation}
    which means that
    \begin{equation}
         [2,4]\backslash \left(\bigcup_{1\le i,j\le t} R_{ij}\right) \not=\varnothing,
    \end{equation}
    and any element $m$ of this set satisfies \eqref{eq:disjoint}.
\end{proof}
\subsection{Proof of Lemma~\ref{lemma:primedist}}
\label{sec:primedistproof}
\begin{proof}
    When $i=j$, it is straightforward to see $\left||\theta_i-\theta_j|-\frac{k}{p_l}\right|\geq\frac{1}{2p_{\frac{t(t-1)}{2}}p_{\frac{t(t-1)}{2}+1}}$ when $p_l\nmid k$, and the case $t=1$ is thus verified. In the following, we assume $t\ge2$ and prove $\underset{\substack{1\leq i< j\leq t\\p_l\nmid k}}{\min}~ \left||\theta_i-\theta_j|-\frac{k}{p_l}\right|\geq\frac{1}{2p_{\frac{t(t-1)}{2}}p_{\frac{t(t-1)}{2}+1}}$ by contradiction. Suppose there is some $(\theta_1, \ldots, \theta_t)\in\RR^t$ such that \eqref{eq:primeineq} does not hold. Then for any $p_l\in\{p_1, p_2, \ldots, p_{\frac{t(t-1)}{2}+1}\}$, there are some $(i_l, j_l)$ such that $\left||\theta_i-\theta_j|-\frac{k_l}{p_l}\right|<1/({2p_{\frac{t(t-1)}{2}}p_{\frac{t(t-1)}{2}+1}})$ for some $p_l\nmid k_l$. Since there are at most $\frac{t(t-1)}{2}$ different values that $|\theta_i-\theta_j|$ can take, but there are $\frac{t(t-1)}{2}+1$ different prime numbers $p_l$, there must be some $l$ and $l'$ with $l\not=l'$ such that $|k_l/p_l-k_{l'}/p_{l'}|<1/({p_{\frac{t(t-1)}{2}}p_{\frac{t(t-1)}{2}+1}})$. Hence, $k_l/p_l=k_{l'}/p_{l'}$, otherwise $|k_l/p_l-k_{l'}/p_{l'}|\geq 1/p_lp_{l'}\geq 1/({p_{\frac{t(t-1)}{2}}p_{\frac{t(t-1)}{2}+1}})$, which contradicts with $|k_l/p_l-k_{l'}/p_{l'}|<1/({p_{\frac{t(t-1)}{2}}p_{\frac{t(t-1)}{2}+1}})$. But $p_l$ and $p_{l'}$ are different prime numbers, so $k_l/p_l=k_{l'}/p_{l'}$ if and only if $k_l=mp_l$ and $k_{l'}=mp_{l'}$ for some integer $m$, which contradicts with $p_l\nmid k_l$ and $p_{l'}\nmid k_{l'}$. The contradiction indicates that \eqref{eq:primeineq} must hold. 
\end{proof}
\subsection{Proof of Lemma~\ref{thm:intmk}}
\label{sec:intmkproof}
\begin{proof}
Since $\norms{E_{\ell-1}}\leq\frac{S\eta}{M_{\ell-1}}$, one has $\norms{M_{\ell-1}E_{\ell-1}}\leq{S\eta}$ and $\norms{M_{\ell-1}B_\TT(E_{\ell-1}, \frac{\eta}{M_{\ell-1}})} = \norms{B_\TT(M_{\ell-1}E_{\ell-1}, \eta)}\leq 3S\eta$. Let $B_\TT(M_{\ell-1}E_{\ell-1}, \eta) = \cup_{1\leq i\leq t} [\theta_i-\eta_i, \theta_i+\eta_i]$, then $\sum_i\eta_i\leq\frac{3}{2}S\eta$. According to \Cref{lemma:primedist}, one can choose $m_\ell \in \{p_1, p_2, \ldots, p_{\frac{S_{\ell-1}(S_{\ell-1}-1)}{2}+1}\}$ such that
\[
\begin{aligned}
\norms{\theta_i-\theta_j-\frac{q}{m_\ell }}&\geq\frac{1}{2p_{\frac{S_{\ell-1}(S_{\ell-1}-1)}{2}}p_{\frac{S_{\ell-1}(S_{\ell-1}-1)}{2}+1}}
>\frac{3S\eta}{2}\geq \eta_i+\eta_j, \quad \text{if }m_\ell \nmid q.
\end{aligned}
\]
This implies that
\[
\begin{aligned}
&\left(B_\TT(M_{\ell-1}E_{\ell-1}, {\eta})+\frac{q}{m_\ell }\right)\cap B_\TT(M_{\ell-1}E_{\ell-1}, {\eta}) = \varnothing \mod 1,\quad \text{if }m_\ell \nmid q,
\end{aligned}
\]
which is equivalent to 
\begin{equation}\label{eq:stepmk}
\begin{aligned}
&\left(B_\TT(E_{\ell-1}, \frac{\eta}{M_{\ell-1}})+\frac{q}{m_\ell M_{\ell-1}}\right)\cap B_\TT(E_{\ell-1}, \frac{\eta}{M_{\ell-1}}) = \varnothing \mod 1,\quad \text{if }m_\ell \nmid q.
\end{aligned}
\end{equation}
Now we prove \eqref{eq:stepmkpruned} inductively. For $\ell=1$, it is implied by \eqref{eq:stepmk} since $M_0=1$. Suppose \eqref{eq:stepmkpruned} already holds for $\ell-1$, then from \Cref{lem:inclusion2} we know $E_{\ell-1}\subset B_\TT(E_{\ell -2},\frac{\eta}{2M_{\ell -2}})$, and thus $B_\TT(E_{\ell-1},\frac{\eta}{M_{\ell-1}})\subset B_\TT(E_{\ell -2},\frac{\eta}{M_{\ell-1}}+\frac{\eta}{2M_{\ell -2}})\subset B_\TT(E_{\ell -2},\frac{\eta}{M_{\ell -2}})$.
Therefore, from the induction hypothesis, one can deduce that
\begin{equation}\label{eq:spcase}
\begin{aligned}
&\left(B_\TT(E_{\ell-1}, \frac{\eta}{M_{\ell-1}})+\frac{q'}{M_{\ell-1}}\right)\cap B_\TT(E_{\ell-1}, \frac{\eta}{M_{\ell-1}}) = \varnothing\mod 1,\quad \text{if }M_{\ell-1}\nmid q',
\end{aligned}
\end{equation}
where we used $m_{\ell-1}M_{\ell -2} = M_{\ell-1}$. If we let $q=m_\ell q'$ in \eqref{eq:spcase}, then $M_{\ell-1}\nmid q'$ is equivalent to $m_\ell \mid q,\ (m_\ell M_{\ell-1})\nmid q$, so \eqref{eq:stepmkpruned} is proved combining \eqref{eq:stepmk} and \eqref{eq:spcase}.
\end{proof}

\subsection{Proof of Theorem~\ref{thm:complexityreal 2}}
\label{sec:complexityreal 2proof}
\begin{proof}
  We will prove that \Cref{alg:main} will work at each step $\ell$. The procedure is almost identical to \Cref{thm:complexityreal}. We only need to check that after step $\ell-1$ we can choose an $m_\ell\in[2,4]$ that satisfies both \eqref{eq:aug} and the gap $\Delta_\ell:= \min_{1\le i<j\le S}{|M_\ell\lambda_i-M_\ell\lambda_j|_s}$ is greater than $\tilde{\Delta}$ so that ESPRIT can work within the error bound given in \Cref{thm:esprit}. In step $\ell -1$ of \Cref{alg:main}, we can already bound $\lambda_i$ in a interval with length less than $\eta/M_{\ell-1}$, which can be denoted by $\lambda_i\in[\theta_i-\frac{\eta}{2M_{\ell-1}},\theta_i+\frac{\eta}{2M_{\ell-1}}]$. Therefore, $\Delta_\ell \ge \tilde{\Delta}$ is equivalent to make the $S$ intervals
    \begin{equation}
        \left\{\left[\theta_i-\frac{\eta}{2M_{\ell-1}}-\frac{\tilde{\Delta}}{m_\ell M_{\ell-1}},\theta_i+\frac{\eta}{2M_{\ell-1}}+\frac{\tilde{\Delta}}{m_\ell M_{\ell-1}}\right]\right\}_{i=1}^S
    \end{equation}
    disjoint in modulo-$\frac{1}{M_{\ell-1}m_\ell}$ sense. Since $m_\ell\ge 2$, we can relax the term $\frac{\tilde{\Delta}}{m_\ell M_{\ell-1}}$ as $\frac{\tilde{\Delta}}{2M_{\ell-1}}$.
    Note that \eqref{eq:aug} is equivalent to say that the $S$ intervals
    \begin{equation}
        \left[\theta_i-\frac{3\eta}{4M_{\ell-1}},\theta_i+\frac{3\eta}{4M_{\ell-1}}\right]\quad (1\le i\le S)
    \end{equation}
    are disjoint in modulo-$\frac{1}{M_{\ell-1}m_\ell}$ sense. Therefore, we conclude that we can find a proper $m_\ell$ by letting $M=M_{\ell-1}$ and $\zeta_i = \max\{\frac{\eta+\tilde{\Delta}}{2}, \frac{3\eta}{4}\}$ in \Cref{lemma:m_real}, since $\sum_{i=1}^S\zeta_i = S\max\{\frac{\eta+\tilde{\Delta}}{2}, \frac{3\eta}{4}\}\le\frac{1}{8S(2S-1)}$.

    After proving that the ESPRIT algorithm described in \Cref{sec:spgap} can be used in \Cref{alg:main} given the existence of the spectral gap, the rest of the proof of the complexity bounds are the same as \Cref{thm:complexityreal}.
\end{proof}

\subsection{Proof of Theorem~\ref{thm:complexityint 2}}
\label{sec:thm:complexityint 2proof}

\begin{proof}
    From ESPRIT's assumptions and properties, $E_1$ already has $S$ disjoint intervals, each containing an actual spike. All steps in the proof of \Cref{thm:intmk} go through except that one needs to check the assumptions for \Cref{thm:esprit} for each iteration in \Cref{alg:main}. In other words, one must check that the minimum separation among $\La_\ell$ is bounded from below by $\tilde{\Delta}$. We prove this by enhancing the induction hypothesis in the proof of \Cref{thm:intmk} with an additional condition:
    \begin{equation}
        \min_{1\le s\le s'\le S, ~n\in \ZZ} |M_{\ell}(\lambda_s-\lambda_{s'})-n| \ge \tilde{\Delta}.
    \end{equation}
    When $\ell=0$, this follows from the definition of $\tilde{\Delta}$. Now assume that this holds for $\ell-1$. By the choice of $m_\ell$ in the proof of \Cref{thm:intmk}, one has
    \[
    \min_{\substack{1\le s\le s'\le S\\q\in \ZZ, m_\ell\nmid q}} \left|\theta_s-\theta_{s'}-\frac{q}{m_\ell}\right| \ge \frac{1}{2p_{\half{S(S-1)}}p_{\half{S(S-1)}+1}},
    \]
    where $\theta_s$ is the center of the $s$-th interval in $E_{\ell-1}$, which is also the $s$-th element in $\tilde{\Lambda}_{\ell-1}$ defined in \Cref{thm:esprit}. By the property of $Y_\ell$, one has
    \[
    \begin{aligned}
    &\min_{\substack{1\le s\le s'\le S\\q\in \ZZ, m_\ell\nmid q}} \left|M_{\ell-1}(\lambda_s-\lambda_{s'})-\frac{q}{m_\ell}\right| \ge \min_{\substack{1\le s\le s'\le S\\q\in \ZZ, m_\ell\nmid q}} \left|\theta_s-\theta_{s'}-\frac{q}{m_\ell}\right| -S\eta\ge\frac{1}{4p_{\half{S(S-1)}}p_{\half{S(S-1)}+1}}.
    \end{aligned}
    \]
    Thus
    \begin{equation}\label{eq:gapstep}
    \begin{aligned}
    \min_{\substack{1\le s\le s'\le S\\q\in \ZZ, m_\ell\nmid q}} &\left|M_{\ell}(\lambda_s-\lambda_{s'})-{q}\right| \ge \frac{m_\ell}{4p_{\half{S(S-1)}}p_{\half{S(S-1)}+1}} \ge \frac{1}{2p_{\half{S(S-1)}}p_{\half{S(S-1)}+1}}\ge\tilde{\Delta}.
    \end{aligned}
    \end{equation}
    By the induction hypothesis, one also has
    \[
    \min_{1\le s\le s'\le S, ~n\in \ZZ} |M_{\ell-1}(\lambda_s-\lambda_{s'})-n| \ge \tilde{\Delta},
    \]
    which means 
    \[
    \min_{1\le s\le s'\le S, ~q'\in \ZZ} |M_{\ell}(\lambda_s-\lambda_{s'})-m_\ell q'| \ge m_\ell\tilde{\Delta} > \tilde{\Delta},
    \]
    and combining with \eqref{eq:gapstep} one obtains
    \[
    \min_{1\le s\le s'\le S, ~n\in \ZZ} |M_{\ell}(\lambda_s-\lambda_{s'})-n| \ge  \tilde{\Delta}.
    \]
    The other steps in the proof of \Cref{thm:intmk} can be used directly to show that the algorithm works, and the arguments for the complexity bounds are the same as the ones in \Cref{thm:complexityint}.  
\end{proof}

\subsection{Proof of Theorem~\ref{thm:real hybrid}}
\label{sec:real hybridproof}

\begin{proof}
    Here, we only need to take care of the choice of the sequence of $\{m_\ell\}$, since the rest of the proof is the same as \Cref{thm:complexityreal} and \Cref{thm:complexityreal 2}. According to \Cref{lemma:m_real}, we need to choose $m_\ell$ such that \eqref{eq:aug} holds. In addition, when $M_\ell>\tilde{M}$, we also have to guarantee the spectral gap $\Delta_\ell = \min_{1\le i<j\le S}{|M_\ell\lambda_i-M_\ell\lambda_j|_s}$ is greater than $\tilde{\Delta}$ to meet the requirements of ESPRIT in \Cref{thm:esprit}.

    At the stage of choosing $m_\ell$, we know that $E_{\ell-1} = \bigcup_{i=1}^{S_{\ell-1}}I_{\ell -1,i}$ is the union of at most $S$ disjoint intervals and $\abs{E_{\ell-1}} \le \frac{S\eta}{M_{\ell-1}}$. Therefore, its neighborhood $B\left(E_{\ell-1},\frac{\max\{\tilde{\Delta}, \eta\}}{2M_{\ell-1}}\right)$ must be $t$ disjoint intervals with $t\le S$ and $\abs{B\left(E_{\ell-1},\frac{\max\{\tilde{\Delta}, \eta\}}{2M_{\ell-1}}\right)}\le \frac{2S\max\{\tilde{\Delta}, \eta\}}{M_{\ell-1}}$. This is exactly the case in \Cref{lemma:m_real} when taking $G = B\left(E_{\ell-1},\frac{\max\{\tilde{\Delta}, \eta\}}{2M_{\ell-1}}\right)$, $M=M_{\ell-1}$ and $\zeta=\max\{\tilde{\Delta}, \eta\}$. Therefore, the conclusion of \Cref{lemma:m_real} guarantees that we can constructively find an $m_\ell$ such that 
\begin{equation}\label{eq:real seperate}
\begin{aligned}
    &\left(B\left(E_{\ell-1},\frac{\max\{\tilde{\Delta}, \eta\}}{2M_{\ell-1}}\right) + \frac{q}{M_{\ell-1}m_\ell}\right)\cap B\left(E_{\ell-1},\frac{\max\{\tilde{\Delta}, \eta\}}{2M_{\ell-1}}\right) = \varnothing, \quad \forall q\in\ZZ\backslash\{0\},
\end{aligned}
\end{equation} 
which is stronger than \eqref{eq:aug}. 

Moreover, \eqref{eq:real seperate} means that $\abs{\lambda_i+ \frac{q}{M_{\ell}}-\lambda_j}\ge\frac{\tilde{\Delta}}{M_{\ell-1}}$ for $\lambda_i,\ \lambda_j\in \Lambda$ and $q\in\ZZ\backslash\{0\}$, since $\Lambda\in E_{\ell-1}$. Therefore, it holds that $\abs{M_\ell\lambda_i-M_{\ell}\lambda_j+q}\ge\tilde{\Delta}$ for $q\in\ZZ\backslash\{0\}$. When $M_\ell>\tilde{M}$, we also have $\abs{M_\ell\lambda_i-M_{\ell}\lambda_j}\ge M_\ell \Delta\ge \tilde{\Delta}$, which means $$\Delta_\ell = \min_{1\le i<j\le S}{|M_\ell\lambda_i-M_\ell\lambda_j|_s}\ge \tilde{M}\Delta=\tilde{\Delta},$$
and this meets the requirement of ESPRIT.

At the stage of $M_\ell\le \tilde{M}$, we have reached the accuracy of $\eps_1 = \tilde{M}^{-1}$, so \Cref{thm:complexityreal} gives that the maximal runtime is
\[
T_{\mathrm{max,1}} = \mo{\eta K_1\eps_1^{-1}} = \mo{\log\left(\frac{1}{\beta}\right)\eta^{-1}\tilde{\Delta}\Delta^{-1}},
\]
and the total runtime is 
\[
T_{\mathrm{total,1}} = \mo{\eta K_1^2\eps_1^{-1}N_{HR,1}} = \tmo{\eta^{-2}\beta^{-2}\tilde{\Delta}\Delta^{-1}}.
\]
At the stage of $M_\ell> \tilde{M}$, according to \Cref{thm:complexityreal 2} the maximal runtime is
\[
T_{\mathrm{max,2}} = \mo{\eta K_2\eps^{-1}} = \mo{\eta\tilde{\Delta}^{-1}\eps^{-1}},
\]
and the total runtime is
\[
T_{\mathrm{total,2}} = \mo{\eta K_2^2\eps^{-1}N_{HR,2}} = \tmo{\eta\omega^{-2}\tilde{\Delta}^{-2}\eps^{-1}}.
\]
Then we can conclude by $\tm = \max\{T_{\mathrm{max,1}}, T_{\mathrm{max,2}}\}$ and $\tt = T_{\mathrm{total,1}} +T_{\mathrm{total,2}}$.
\end{proof}
\subsection{Proof of Theorem~\ref{thm:int hybrid}}
\label{sec:int hybridproof}
\begin{proof}
    Let $\tilde{\ell}$ be the largest $\ell$ such that $M_{\ell}\le\tdm$. From the definition of $N_{\HR,1}$ above, one knows from Hoeffding's inequality that for any $\ell\le{\tilde{\ell}}$ and $|k|\le K_1$, 
    \begin{equation}
	\PP\left(\abs{{y}_\ell(k) - \hat{f}_\ell(k)} < \alpha\right) > 1-\frac{\p}{(\left\lceil\log_2\frac{\eta}{\eps}\right\rceil+1)(K_1+1)},
  \end{equation}
  and for any $\ell>{\tilde{\ell}}$ and $|k|\le K_2$, 
    \begin{equation}
	\PP\left(\abs{{y}_\ell(k) - \hat{f}_\ell(k)} < \alpha\right) > 1-\frac{\p}{(\left\lceil\log_2\frac{\eta}{\eps}\right\rceil+1)(K_2+1)},
  \end{equation}
  Thus $\abs{{y}_\ell(k) - \hat{f}_\ell(k)} < \alpha$ is true for all $\ell$ and $k$ with probability at least $1-\p$ by the union bound. We condition on this event in the rest of the proof.

  We first show that the spectral gap can be sufficiently enlarged by applying a burn-in process using the method in \Cref{sec:ns}. The proof is similar to that of \Cref{thm:complexityint 2}. We enhance the induction hypothesis in \Cref{thm:intmk} by an additional condition
    \begin{equation}
        \Delta_\ell \ge \min\left\{M_\ell\Delta, \frac{1}{2p_{\frac{S(S-1)}{2}}p_{\frac{S(S-1)}{2}+1}}\right\}.
    \end{equation}
    When $\ell=0$, this follows from the definition of $\Delta_\ell$. Now assume that this holds for $\ell-1$. By the choice of $m_\ell$ in the proof of \Cref{thm:intmk}, one has
    \[
    \min_{\substack{1\le s\le s'\le t\\q\in \ZZ, m_\ell\nmid q}} \left|\theta_s-\theta_{s'}-\frac{q}{m_\ell}\right| \ge \frac{1}{2p_{\half{S(S-1)}}p_{\half{S(S-1)}+1}},
    \]
    where $\theta_s$ is the center of the $s$-th interval in $E_{\ell-1}$, and $t$ is the number of intervals in $E_{\ell-1}$. For any $\lambda$ and $\lambda'$ in $\La$, if $\lambda$ and $\lambda'$ belong to different intervals indexed by $s$ and $s'$, one has
    \[
    \begin{aligned}
    &\min_{q\in \ZZ, m_\ell\nmid q} \left|M_{\ell-1}(\lambda-\lambda')-\frac{q}{m_\ell}\right| \ge \min_{\substack{1\le s\le s'\le t\\q\in \ZZ, m_\ell\nmid q}} \left|\theta_s-\theta_{s'}-\frac{q}{m_\ell}\right| -S\eta\ge\frac{1}{4p_{\half{S(S-1)}}p_{\half{S(S-1)}+1}},
    \end{aligned}
    \]
    by the property of $Y_\ell$.
    Thus
    \begin{equation}\label{eq:lbl1}
    \begin{aligned}
    \min_{q\in \ZZ, m_\ell\nmid q}& \left|M_{\ell}(\lambda-\lambda')-{q}\right| \ge \frac{m_\ell}{4p_{\half{S(S-1)}}p_{\half{S(S-1)}+1}} \ge\min\left\{M_\ell\Delta, \frac{1}{2p_{\frac{S(S-1)}{2}}p_{\frac{S(S-1)}{2}+1}}\right\}.
    \end{aligned}
    \end{equation}
    On the other hand, if $\lambda$ and $\lambda'$ belong to the same interval, then $|M_{\ell-1}(\lambda-\lambda')|\le S\eta$ and $|M_\ell(\lambda-\lambda')|\le m_\ell S\eta \le 1/4$, so
    \begin{equation}\label{eq:lbl2}
    \begin{aligned}
    \min_{q\in \ZZ, m_\ell\nmid q} &\left|M_{\ell}(\lambda-\lambda')-{q}\right| = |M_\ell(\lambda-\lambda')| \ge M_\ell \Delta \ge \min\left\{M_\ell\Delta, \frac{1}{2p_{\frac{S(S-1)}{2}}p_{\frac{S(S-1)}{2}+1}}\right\}.
    \end{aligned}
    \end{equation}
    By the induction hypothesis, one also has
    \[
    \min_{n\in \ZZ} |M_{\ell-1}(\lambda-\lambda')-n| \ge  \min\left\{M_{\ell-1}\Delta, \frac{1}{2p_{\frac{S(S-1)}{2}}p_{\frac{S(S-1)}{2}+1}}\right\},
    \]
    which means 
    \[
    \min_{q'\in \ZZ} |M_{\ell}(\lambda-\lambda')-m_\ell q'| \ge \min\left\{M_{\ell}\Delta, \frac{1}{2p_{\frac{S(S-1)}{2}}p_{\frac{S(S-1)}{2}+1}}\right\},
    \]
    and combining with \eqref{eq:lbl1} and \eqref{eq:lbl2} one obtains
    \[
    \min_{n\in \ZZ} |M_{\ell}(\lambda-\lambda')-n| \ge \min\left\{M_{\ell}\Delta, \frac{1}{2p_{\frac{S(S-1)}{2}}p_{\frac{S(S-1)}{2}+1}}\right\}.
    \]

    When $\ell=\tilde{\ell}+1$, one has $\min_{n\in \ZZ} |M_{\ell}(\lambda-\lambda')-n| \ge \min\left\{M_{\ell}\Delta, \frac{1}{2p_{\frac{S(S-1)}{2}}p_{\frac{S(S-1)}{2}+1}}\right\}\ge\tilde{\Delta}$. With the same proof as \Cref{thm:complexityint 2}, one can show that for any $\ell>\tilde{\ell}$, $E_\ell$ has $S$ disjoint intervals, each containing an actual spike and that
    \[
    \min_{1\le s\le s'\le S, ~n\in \ZZ} |M_{\ell}(\lambda_s-\lambda_{s'})-n| \ge  \tilde{\Delta},
    \]
    which ensures that the ESPRIT subroutine can be applied for any $\ell>\tilde{\ell}$.

In the first stage, i.e., when $M_\ell\le \tilde{M}$, the maximal runtime is
\[
T_{\mathrm{max,1}} = \mo{M_{\tilde{\ell}}K_1} = O_{S}\left(\eta^{-1}\log\left(\frac{1}{\beta}\right)\tilde{\Delta}\Delta^{-1}\right),
\]
and the total runtime is 
\[
T_{\mathrm{total,1}} = \mo{M_{\tilde{\ell}}K_1^2N_{HR,1}} = O_{S}\left(\eta^{-2}\log^2\left(\frac{1}{\beta}\right)\tilde{\Delta}\Delta^{-1}N_{HR,1}\right).
\]
For the second stage, i.e.,when $M_\ell> \tilde{M}$, according to \Cref{thm:complexityint 2} the maximal runtime is
\[
T_{\mathrm{max,2}} = \mo{\eta K_2\eps^{-1}} = \mo{\eta\tilde{\Delta}^{-1}\eps^{-1}},
\]
and the total runtime is
\[
T_{\mathrm{total,2}} = \mo{\eta K_2^2\eps^{-1}N_{HR,2}} = \mo{\eta\tilde{\Delta}^{-2}\eps^{-1}N_{HR,2}}.
\]
Thus, the overall maximum runtime is
\[
\tm = \tilde{O}_S\left(\max\{\eta^{-1}\tilde{\Delta}\Delta^{-1}, \eta\tilde{\Delta}^{-1}\eps^{-1}\}\right),
\]
and the overall total runtime is
\[
\tt = \tilde{O}_S\left(\eta^{-2}\tilde{\Delta}\Delta^{-1}N_{HR,1}+\eta\tilde{\Delta}^{-2}\eps^{-1}N_{HR,2}\right).
\]
\end{proof}

\end{document}